\documentclass{amsart}

\usepackage{graphicx}
\usepackage{amssymb}
\usepackage{amsmath}
\usepackage{amsfonts}
\usepackage[obeyFinal,textsize=footnotesize]{todonotes}

\usepackage{array}
\usepackage{multirow}
\usepackage{makecell}
\usepackage{arydshln}
\usepackage{multirow,bigdelim}
\usepackage{mathtools}
\usepackage{verbatim}
\usepackage{tabu}
\usepackage[font=small]{caption}
\usepackage{url}
\usepackage[hidelinks]{hyperref}
\usepackage{hhline}
\usepackage{fancyvrb}
\usepackage{comment}
\usepackage{hhline}
\usepackage{algorithm,algpseudocode}
\usepackage{setspace}

\newtheorem{theorem}{Theorem}[section]
\newtheorem{proposition}[theorem]{Proposition}
\newtheorem{corollary}{Corollary}[theorem]

\theoremstyle{definition}

\theoremstyle{remark}

\numberwithin{equation}{section}
\newcommand{\RN}[1]{%
  \textup{\uppercase\expandafter{\romannumeral#1}}%
}
\newcommand{\rn}[1]{%
  {\lowercase\expandafter{(\romannumeral#1)}}%
}

\newcommand{\cA}{\mathcal{A}}
\DeclareMathOperator{\tr}{tr}

\DeclareMathOperator{\Ai}{Ai}
\DeclareMathOperator{\Bess}{Bess}

\newcommand{\twoone}[2]{\begin{bmatrix}#1 \\ #2\end{bmatrix}}
\newcommand{\twotwo}[4]{\begin{bmatrix}#1 & #2\\#3 & #4\end{bmatrix}}

\newcommand{\rt}{\!\!\restriction}
\newcommand{\Ks}{K_{\sin}}
\newcommand{\Kai}{K_{\Ai}}
\newcommand{\Kbess}[1]{K_{\Bess, #1}}
\newcommand{\prob}[1]{\mathbb{P}(\text{#1})}

\makeatletter
\renewcommand*\env@matrix[1][*\c@MaxMatrixCols c]{%
  \hskip -\arraycolsep
  \let\@ifnextchar\new@ifnextchar
  \array{#1}}
\makeatother

\begin{document}

\title[Conditional DPP \& New determinantal expressions]{The conditional DPP approach to random matrix distributions}

\author{Alan Edelman}
\address{Department of Mathematics and Computer Science \& AI Laboratory, Massachusetts Institute of Technology, Cambridge, Massachusetts}
\email{edelman@mit.edu}

\author{Sungwoo Jeong}
\address{Department of Mathematics, Cornell University, Ithaca, New York}
\email{sjeong@cornell.edu}

\begin{abstract}
We present the conditional determinantal point process (DPP) approach to obtain new (mostly Fredholm determinantal) expressions for various eigenvalue statistics in random matrix theory. It is well-known that many (especially $\beta=2$) eigenvalue $n$-point correlation functions are given in terms of $n\times n$ determinants, i.e., they are continuous DPPs. We exploit a derived kernel of the conditional DPP which gives the $n$-point correlation function conditioned on the event of some eigenvalues already existing at fixed locations. 

\par Using such kernels we obtain new determinantal expressions for the joint densities of the $k$ largest eigenvalues, probability density functions of the $k^\text{th}$ largest eigenvalue, density of the first eigenvalue spacing, and more. Our formulae are highly amenable to numerical computations and we provide various numerical experiments. Several numerical values that required hours of computing time can now be computed in seconds with our expressions, which proves the effectiveness of our approach.

\par We also demonstrate that our technique can be applied to an efficient sampling of DR paths of the Aztec diamond domino tiling. Further extending the conditional DPP sampling technique, we sample Airy processes from the extended Airy kernel. Additionally we propose a sampling method for non-Hermitian projection DPPs. 
\end{abstract}

\maketitle

\section{Introduction}

\subsection{The Conditional DPP Method}

This paper shows that conditional determinantal point processes (DPP) can be exploited in novel ways to create interesting expressions and yield highly efficient algorithms for computations in random matrix theory and beyond. We call this the \textit{conditional DPP approach}. Figure~\ref{fig:gallery} presents a gallery of examples created by these new expressions.

\begin{figure}[t]
    \centering
    \begin{tabular}{|c|c|}
    \hline
    \multicolumn{2}{|c|}{\makecell[b]{\vspace{-0.3cm} \\ Correlation coefficient (soft-edge) \\ {\large $\rho(\lambda_1, \lambda_2) = 0.50564723159...$} \vspace{0.1cm}\\ 
    Computation time : \textbf{16 hours} (2010, \cite{bornemann2010numerical}) $\longrightarrow$ \textbf{2 minutes} (Section~\ref{sec:klargest}) 
    }}   \\ 
    \hline
    \includegraphics[width=0.45\textwidth]{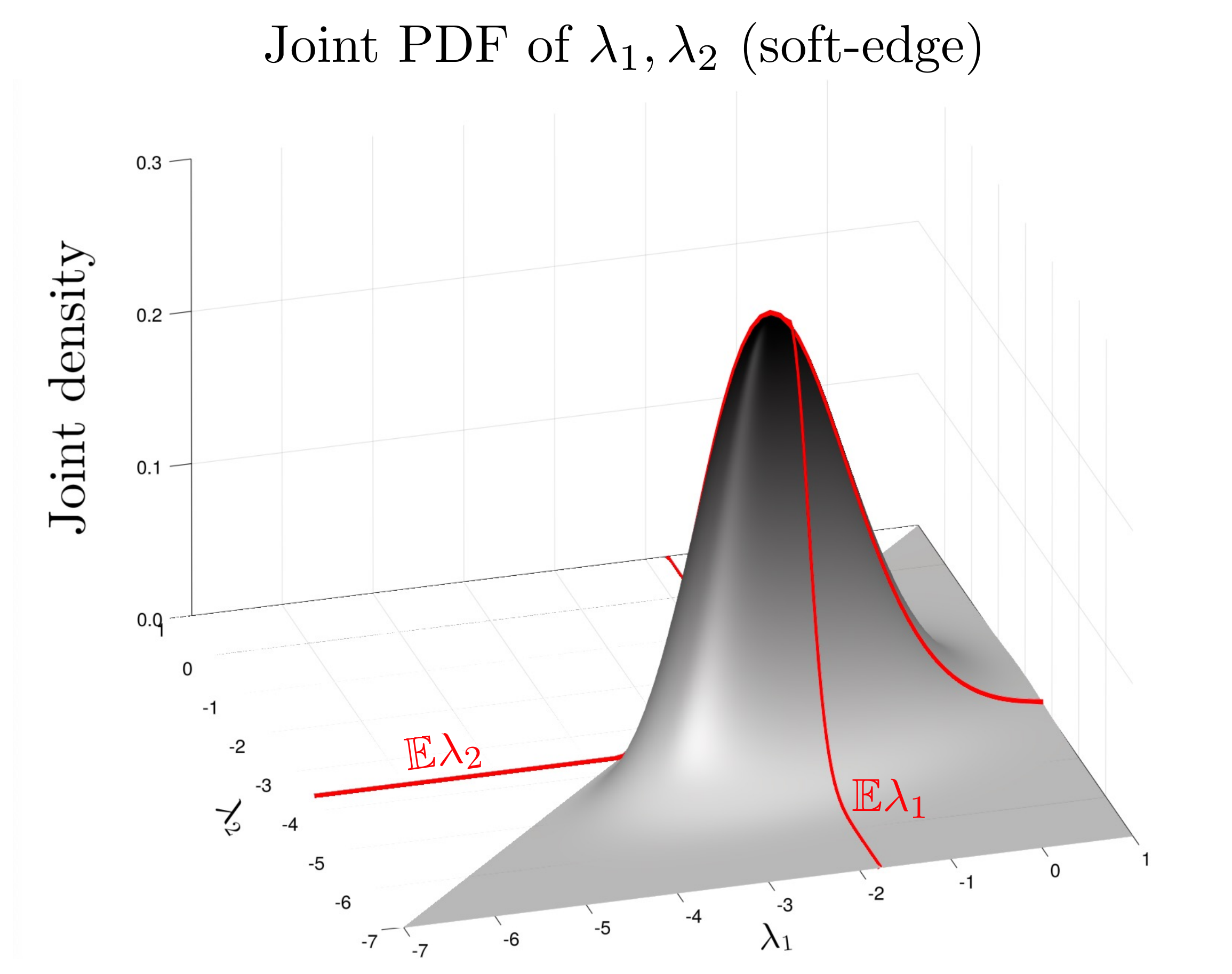}& \includegraphics[width=0.45\textwidth]{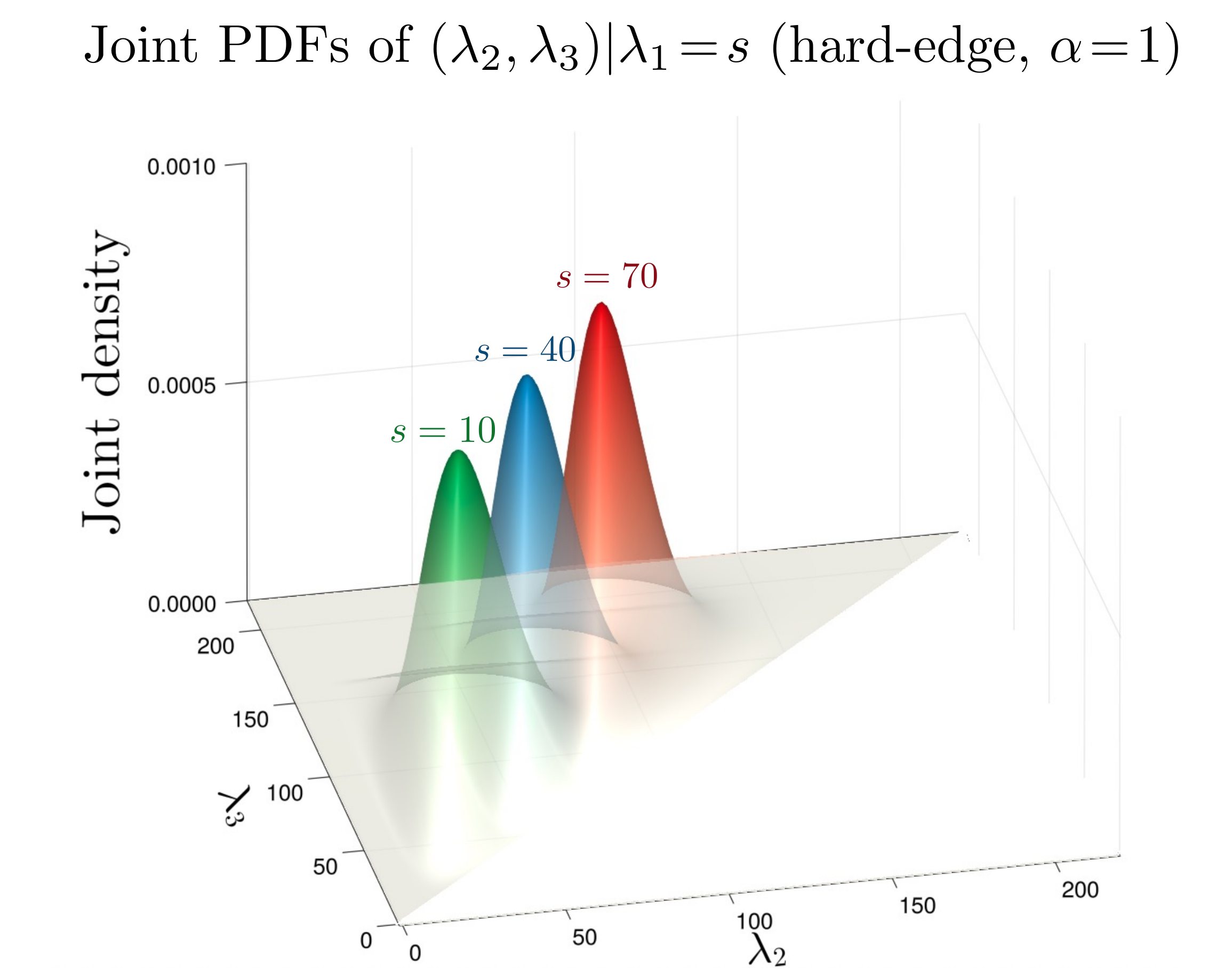} \\ 
    Section~\ref{sec:klargest}, Equation \eqref{eq:twojointpdfsoftedge} & Section~\ref{sec:klargest}\\\hline
    \includegraphics[width=0.45\textwidth]{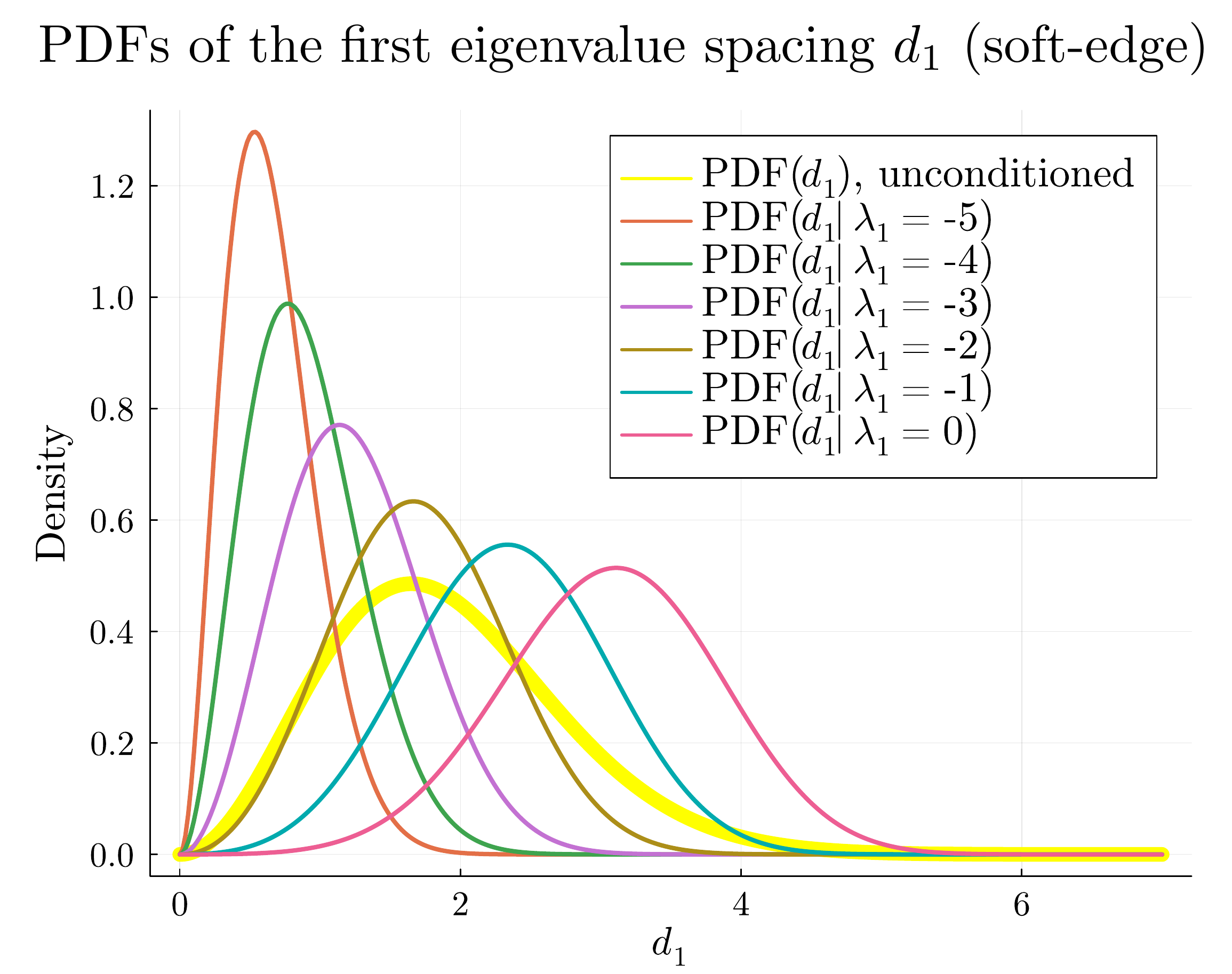} & \includegraphics[width=0.45\textwidth]{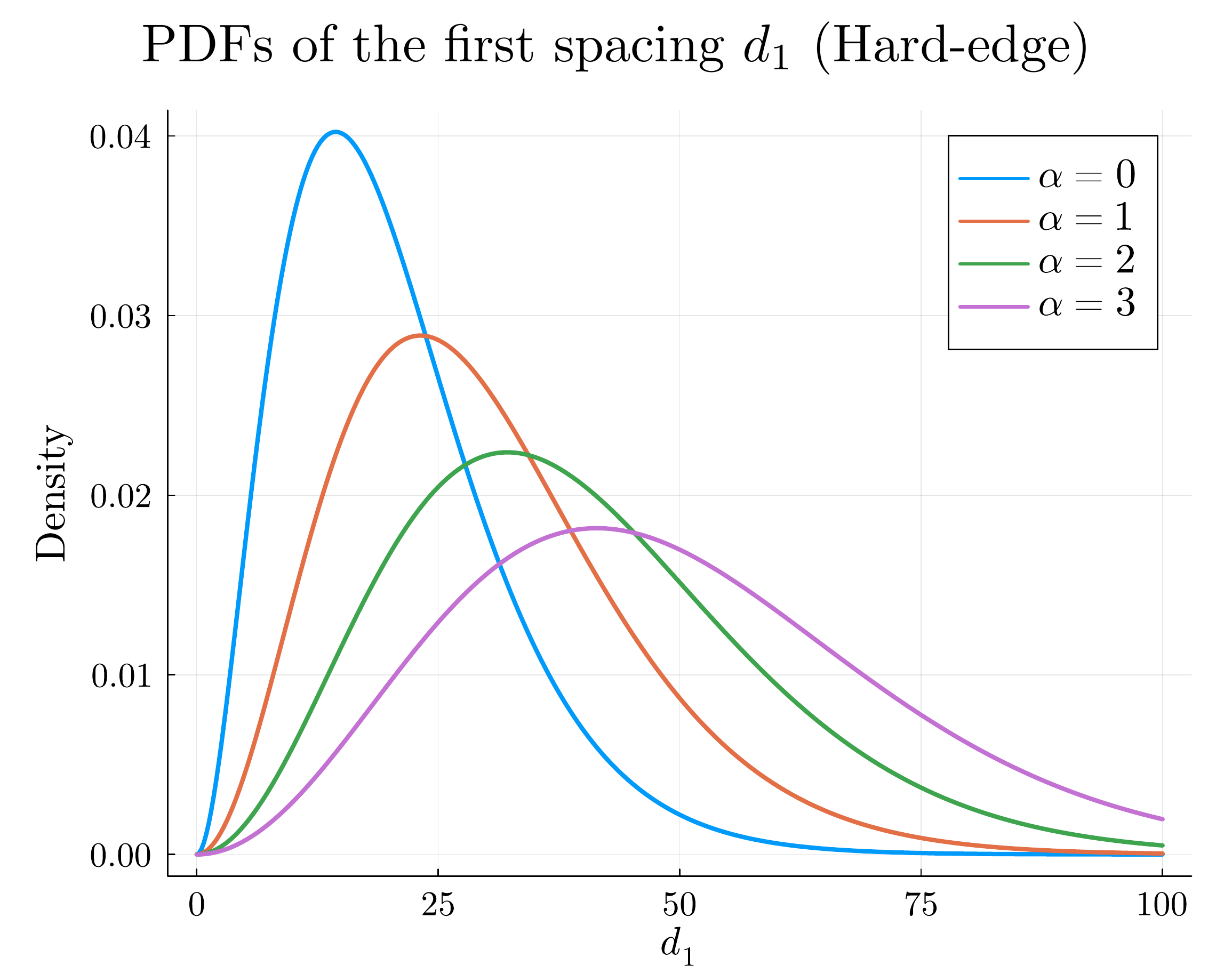} \\ 
    Section~\ref{sec:spacing} & Section~\ref{sec:spacing}, Equation \eqref{eq:spacingpdf} \\ \hline
    \includegraphics[width=0.45\textwidth]{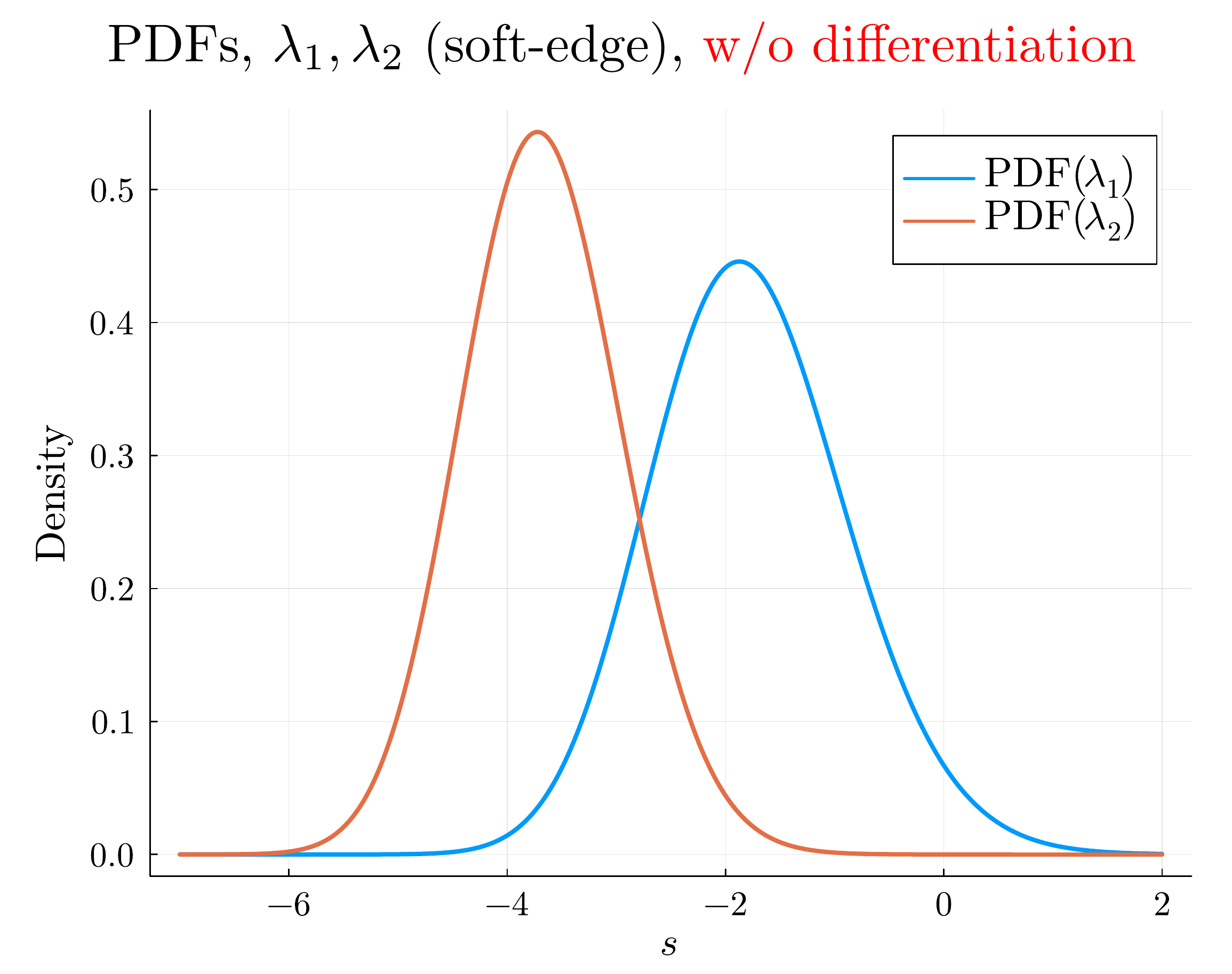} & \includegraphics[width=0.45\textwidth]{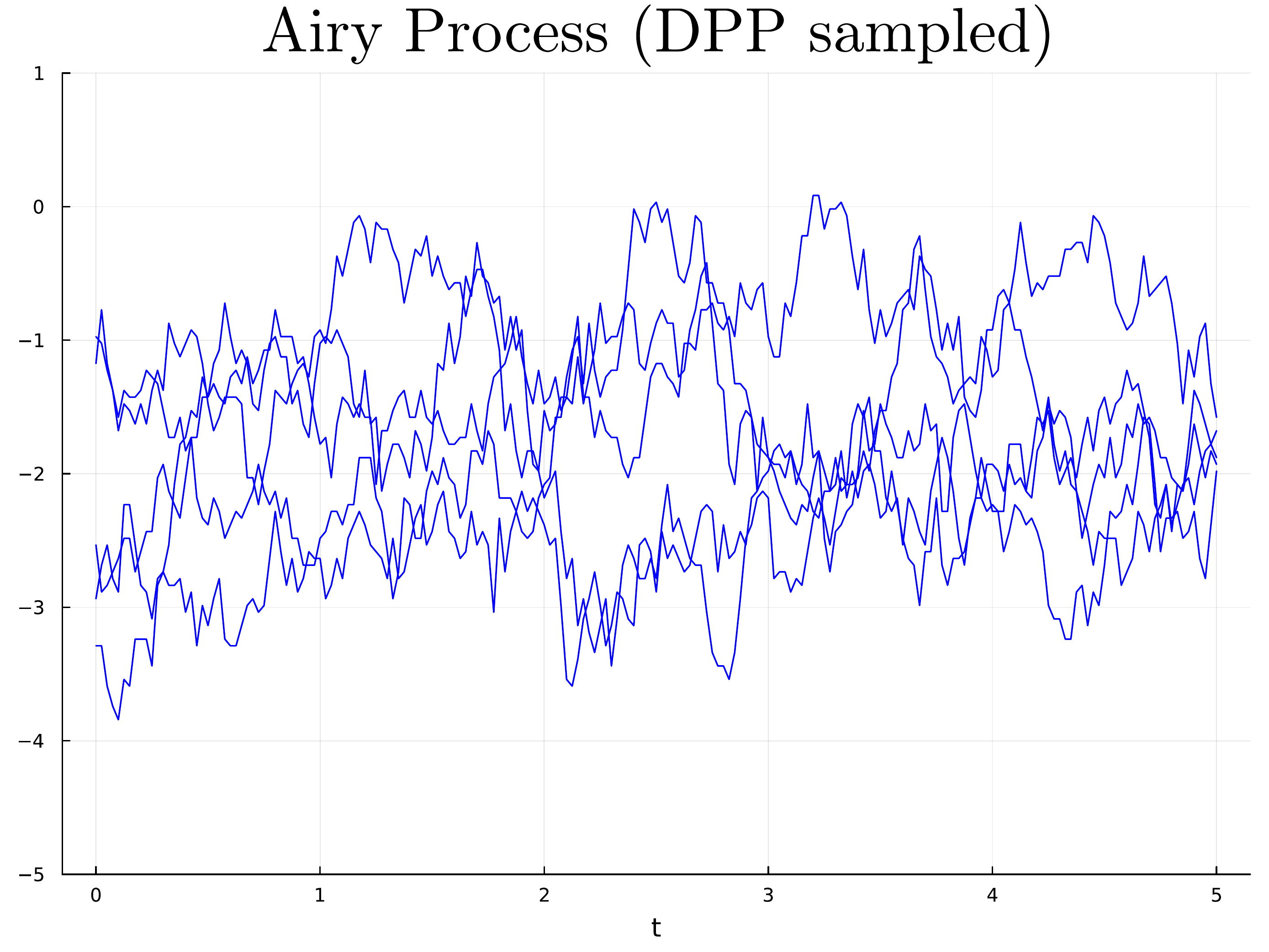}\\
    Sections~\ref{sec:extremepdf}, \ref{sec:2ndeig}, Eq \eqref{eq:tracywidompdf}, \eqref{eq:tw2ndpdf}& Section~\ref{sec:airyprocess}\\
    \hline
    \end{tabular}
    \caption{\normalsize Plots and numerical results enabled by the conditional DPP approach. Note that $\alpha$ is the parameter of the LUE. \label{fig:gallery}}
\end{figure}





\subsection{Technical Background}
Determinantal representations frequently arise in random matrix theory, especially in the study of eigenvalues of $\beta=2$ (complex) random matrices \cite{forrester2010log,mehta2004random}. A number of random matrix $n$-point eigenvalue correlation functions \cite{dyson1962statistical} are given in terms of the following determinantal formula 
\begin{equation}\label{eq:npointdeterminant}
    p(x_1, \dots, x_n) = \det\left(\left[K(x_i, x_j)\right]_{i, j= 1\dots,n}\right),
\end{equation}
with some kernel $K$. A basic example is the $N\times N$ Gaussian unitary ensemble (GUE) and the Hermite kernel $K = K_\text{Herm}^{(N)}$ defined as
\begin{equation}\label{eq:HermiteKernel}
    K_\text{Herm}^{(N)}(x, y) = \sum_{i=0}^{N-1}\phi_i(x)\phi_i(y) = \sqrt\frac{N}{2}\frac{\phi_N(x)\phi_{N-1}(y)-\phi_{N-1}(x)\phi_N(y)}{x-y},
\end{equation}
where $\phi_j(x) = \exp(-\frac{x^2}{2})H_j(x)/(2^j\sqrt\pi j! )^{1/2}$ and $H_j$'s are the Hermite polynomials. 

\par Defined by such determinantal representations, a point process is called a (discrete) \textit{determinantal point process}, if for any subset $S$ of the ground index (point) set $\mathcal{G}$ a random sample $\mathcal{J}$ has the probability, 
\begin{equation}\label{eq:DPPdefinition}
    \prob{$S\subset \mathcal{J}$} = \det\left(\left[K(x, y)\right]_{x, y\in S}\right),
\end{equation}
given a matrix (or a kernel, in continuous cases) $K$. In this paper the kernel $K$ is not restricted to symmetric or Hermitian kernels (matrices).

\par In the discrete (finite) case, an exact sampling algorithm is introduced in \cite{hough2006determinantal} for DPPs with Hermitian kernels. This algorithm uses the fact that any Hermitian DPP is a mixture of projection DPPs (elementary DPPs), and also that projection DPPs have a simple exact sampling algorithm. Other sampling algorithms were also studied, for example see \cite{derezinski2019exact,kulesza2011k}, but mostly for DPPs with Hermitian kernels. 

\par However recently a new ``greedy" type algorithm was introduced \cite{launay2020exact,poulson2020high}, based on the successive computations of conditional probabilities\footnote{Such conditional approaches were also considered earlier, e.g., \cite{borodin2005eynard}.} through the block LU decomposition. In each step of this algorithm we determine whether a given index is included in the sample or not by a Bernoulli trial, which we refer to as the \textit{observation}. Depending on the observation at each step, we modify (or keep) a diagonal entry (the pivot of the LU decomposition), then perform a single step of the LU decomposition. This may be less efficient than the standard Hermitian sampler but this algorithm allows one to sample from non-Hermitian DPPs.

\par This paper is inspired by this greedy type algorithm. We make the following three important remarks:
\begin{itemize}\setlength\itemsep{0.4em}
    \item After each observation we obtain a kernel corresponding to the \textbf{new DPP of unobserved points} conditioned on the result of observed points. 
    \item Theoretically, we can \textbf{force specific points} to be included in the sample by just multiplying Bernoulli parameters, instead of random Bernoulli trials. We still obtain the above (conditional) kernel under such an event. 
    \item The observations can be done in an \textbf{arbitrary order} (of point indices).
\end{itemize}
Based on these points, we can use the conditional probability approach to derive several new determinantal expressions of various probability density functions (PDF) and cumulative distribution functions (CDF) in Section \ref{sec:mainresults}.

Not to be understated is the role of algorithms from numerical linear algebra in this work as an inspiration and algorithmic enhancement. One way or another, the kernels of the DPP may undergo a matrix factorization; A potential key to effective algorithms is the recognition of which choice to use when.

\subsection{Main tool}
From the third remark above, let us imagine observing a specific point $s$ first. Then from the second remark, force an eigenvalue at (an infinitesimal interval around) $s$. Finally using the first remark, we introduce the following Proposition which is the key to our results.  
\begin{proposition}\label{prop:conditionalkernel}
Let $K$ be a kernel of an integral operator\footnote{The kernel $K : J\times J \to \mathbb{C}$ is the kernel of some integral operator $\tilde{K}$ on $L^2(J)$ as follows:
\begin{equation*}
    \tilde{K}f(x) = \int K(x,y)f(y)dy.
\end{equation*}
However we will simply denote by $K$ both the kernel and the integral operator since there is no confusion throughout this work.} on $J$ that defines a continuous DPP of the $n$-point correlation function \eqref{eq:npointdeterminant} of some random matrix eigenvalues. Fix a point $s\in J$ and define a derived kernel
\begin{equation}\label{eq:newkernel}
    K^{(s)}(x, y) := K(x, y) - \frac{K(x, s)K(s, y)}{K(s, s)}.
\end{equation}
Then, the $n$-point correlation function $p^{(s)}(x_1, \dots, x_n)$ of the rest of eigenvalues given that an eigenvalue already exists in an infinitesimal interval around $s$ is
\begin{equation*}
    p^{(s)}(x_1,\dots,x_n) = \det\left(\left[K^{(s)}(x_i, x_j)\right]_{i, j=1, \dots, n}\right).
\end{equation*}
In other words, the kernel $K^{(s)}$ defines a continuous DPP \eqref{eq:npointdeterminant} of the eigenvalues conditioned on the event of an eigenvalue existing at $s$. 
\end{proposition}

\par The concept of the conditional DPP in Proposition~\ref{prop:conditionalkernel} can be found (in terms of $L$-ensemble) and justified as a DPP in Borodin and Rains \cite{borodin2005eynard}, where it is proven to be useful when proving the Eynard-Mehta theorem. It has also been discussed in the context of machine learning \cite{kulesza2011k,kulesza2011learning}. 

\par One might notice that the kernel \eqref{eq:newkernel} is in fact the result of a single step of the LU decomposition with the pivot $K(s,s)$, as discussed in the second remark. Note that we can condition on the selection of more than one points (Proposition~\ref{prop:conditionalkernelmpoints}). We emphasize that this kernel is easy-to-use since it is explicit and does not include any infinite summation or differentiation. 

\par Proposition~\ref{prop:conditionalkernel} leads to some new expressions on eigenvalue statistics of random matrices. In Section~\ref{sec:mainresults} we provide several eigenvalue statistics in terms of Fredholm determinants, which are known to be amenable to numerical computation through the method proposed in \cite{bornemann2010numerical,bornemann2009numerical}. These results include but are not limited to:
\begin{itemize}\setlength\itemsep{0.4em}
    \item PDF, CDF of the two extreme eigenvalues (Sections~\ref{sec:extremepdf}, \ref{sec:2ndeig})
    \item Joint PDF, CDF of the $k$ extreme eigenvalues (Section \ref{sec:klargest})
    \item PDF, CDF of the first eigenvalue spacing (Section~\ref{sec:klargest})
\end{itemize}
In these results, a random matrix can be chosen to be any random matrix with a determinantal $n$-point correlation function \eqref{eq:npointdeterminant}, such as the GUE, LUE, JUE, soft-edge scaling, hard-edge scaling, etc.


\subsection{Preview \#1: Joint PDF of the two largest eigenvalues}
\par A good example of our approach is the joint PDF $f^{(\lambda_1, \lambda_2)}$ of the $k=2$ largest eigenvalues\footnote{The choice of $k=2$ is arbitrary and for illustrative purpose. We can obtain a joint PDF of any $k$ largest eigenvalues which we discuss in Section~\ref{sec:klargest}.} $\lambda_1\geq \lambda_2$ of a random matrix. Let us use the soft-edge scaling limit of the GUE as an example. It is expressible in terms of a Fredholm determinant\footnote{Throughout this paper, we simplify notation by denoting the restriction $K\rt_{L^2(J)}$ of the operator $K$ to square integrable functions $L^2(J)$ by $K\rt_J$.} using Proposition~\ref{prop:conditionalkernel},
\begin{equation}\label{eq:twojointpdfsoftedge}
    f^{(\lambda_1, \lambda_2)}(x_1, x_2) = \det\left(\twotwo{K(x_1, x_1)}{K(x_1, x_2)}{K(x_2, x_1)}{K(x_2, x_2)}\right)\cdot \det(I - K^{(x_1,x_2)}\rt_{(x_2, \infty)}),
\end{equation}
for $x_1>x_2$ and vanishes otherwise, where $K=\Kai$ is the Airy kernel \eqref{eq:airykernel} and the kernel $K^{(x_1, x_2)}$ is defined in terms of $K$,
\begin{equation*}
    K^{(x_1, x_2)}(x, y) := K(x, y) - \twoone{K(x, x_1)}{K(x, x_2)}^T\!\!\twotwo{K(x_1, x_1)}{K(x_1, x_2)}{K(x_2, x_1)}{K(x_2, x_2)}^{-1}\!\twoone{K(x_1, y)}{K(x_2, y)}. 
\end{equation*}
Using the above formula \eqref{eq:twojointpdfsoftedge} we were able to compute the correlation coefficient of the two largest eigenvalues at the soft-edge scaling limit
\begin{equation*}
    \rho(\lambda_1, \lambda_2) = 0.505\,647\,231\,59...,
\end{equation*}
up to 11 digits in less than 2 minutes. This correlation coefficient has a previously reported computing time of 16 hours for 11 digits in 2010 \cite{bornemann2010numerical}. The formula can also be generalized to the $k$ largest or, similarly, smallest eigenvalues of any random matrices when the $n$-point correlation function of eigenvalues is given in determinantal manner \eqref{eq:npointdeterminant}. See Section \ref{sec:klargest} for details. 


\subsection{Preview \#2: Determinantal expression for the PDF of the Tracy--Widom distribution}
\par The famous Tracy--Widom distribution PDF is often plotted. Interestingly, as far as we know\footnote{Except for a recent approach suggested in \cite{bornemann2022stirling}. See Section~\ref{sec:extremepdf} for details.}, the last step of the computation involves taking the derivative of the CDF. In this preview we give a direct determinantal expression of the Tracy--Widom PDF that gives the plot in the bottom left part of Figure~\ref{fig:gallery}. The CDF expression that is typically used is $F_2(s) = \det(I-\Kai\rt_{(s, \infty)})$. Our idea is that we first fix a level at $s$ and then compute the probability that nothing lies above $s$ with the conditional DPP kernel $\Kai^{(s)}$. 
\begin{proposition}[PDF of the Tracy--Widom distribution]\label{prop:twpdf}
The probability density function $f_2$ of the largest eigenvalue at the soft-edge is
\begin{equation}\label{eq:tracywidompdf}
    \frac{d}{ds}F_2(s) = f_2(s) = \Kai(s, s)\cdot \det\left(I - \Kai^{(s)}\rt_{(s, \infty)}\right),
\end{equation}
where $F_2$ is the Tracy--Widom distribution (CDF), $\Kai$ is the Airy kernel and 
\begin{equation*}
    \Kai^{(s)}(x, y) = \Kai(x, y)- \frac{(\Ai(x)\Ai'(s)-\Ai(s)\Ai'(x))(\Ai(s)\Ai'(y)-\Ai(y)\Ai'(s))}{(x-s)(y-s)(s\Ai(s)^2-\Ai'(s)^2)}
\end{equation*}
is the derived kernel of the conditional DPP as proposed in \eqref{eq:newkernel}.
\end{proposition}

This again is not limited to the soft-edge scaling limit, but also applicable to other random matrices such as the finite GUE, LUE, hard-edge and more. See Section \ref{sec:extremepdf} for further details. We also provide numerical experiments.

\subsection{Outline of the paper}
\par In Section \ref{sec:DPPtheory}, we review the theory of discrete and continuous DPPs and their sampling algorithms. We propose a hybrid sampling method, Algorithm~\ref{alg:HybridDPP}, for DPPs with non-Hermitian projection kernels, for example the DPP of the Aztec diamond \cite{johansson2005arctic}. We then demonstrate an efficient sampling of the DR paths (see Section~\ref{sec:drpath}) without sampling the whole Aztec diamond. 

\par Random matrix applications are discussed in Section~\ref{sec:mainresults}. In Section~\ref{sec:detexpressions} we review some basic random matrix eigenvalue statistics and the conditional DPP approach. Then we derive several new determinantal representations which are efficiently implemented for numerical computations in later sections. In Section~\ref{sec:extremepdf} we obtain a Fredholm determinant expression for the PDF of extreme eigenvalues, such as the Tracy--Widom distribution. In Section~\ref{sec:2ndeig} we specialize on the second largest eigenvalue and provide new formulae for distribution functions. Section~\ref{sec:klargest} discusses the joint PDF of the $k$ largest (extreme) eigenvalues. Applications of these joint PDFs include the first eigenvalue spacing, correlation coefficient of the two largest eigenvalues, and many more. Finally in Section~\ref{sec:airyprocess} we demonstrate the sampling of Airy processes from the DPP. Throughout Section~\ref{sec:mainresults} we provide extensive numerical experiments, All codes may be found \href{https://github.com/sw2030/RMTexperiments}{online}.

\section{Determinantal point processes}\label{sec:DPPtheory}

\subsection{Discrete and Projection DPPs}
Discrete DPPs have a fairly straightforward definition as we saw from \eqref{eq:DPPdefinition}. In particular if we have a finite sized ground set $G$, the kernel is a (finite) $|G|\times |G|$ matrix $K$, often called the \textit{marginal kernel}. For a given DPP $\mathcal{J}$ the following identities are important:
\begin{gather}\label{eq:traceK}
    \tr(K) = \mathbb{E}(|\mathcal{J}|),\\\label{eq:detImK}
    \det(I-K) = \prob{$|\mathcal{J}|=0$}.
\end{gather}

\par It is known that if the kernel matrix $K$ is a \textit{projection matrix}, a DPP has the following special property: \textit{a DPP with a rank $r$ projection marginal kernel draws a sample with the size exactly $r$, i.e. $|J|=r$ (almost surely when continuous).}

\par This property is a cornerstone of the sampling algorithm introduced in \cite{hough2006determinantal}. Imagine an algorithm that draws samples from a given DPP, where sample points are selected in an unsorted (uniformly permuted) order. The probability $P_i$ that a given index $i\in G$ is picked `first' is the following.
\begin{equation*}
    P_i=\prob{$\mathcal{J}=\{i\}$} + \sum_{\substack{i\in S\\ |S|=2}}\frac{1}{2}\prob{$\mathcal{J}=S$}+\sum_{\substack{i\in S\\ |S|=3}}\frac{1}{3}\prob{$\mathcal{J}=S$}+\cdots = \sum_{i\in S}\frac{1}{|S|}\prob{$\mathcal{J}=S$}.
\end{equation*}
Thinking the other way around, if we know $\{P_i\}$, we can  sample the first point (without worrying about additional sample points) according to the discrete random variable $X$ defined by $\prob{$X=i$} = P_i$. However the probabilities $\{P_i\}$ in general do not have a simpler expression in terms of the entries of the marginal kernel. 

\par Nonetheless, for projection DPPs, (normalized) diagonal entries of $K$ equals the probabilities $\{P_i\}$ as $P_i = K_{ii}/r$ is deduced from \eqref{eq:traceK} together with $\tr(K)=r$. Thus, sampling from a projection DPP begins by drawing a single index point from a categorical random variable with normalized diagonal entries as its distribution. After drawing a first point, one can modify the kernel so one can sample points iteratively as we describe in the following paragraphs.

\begin{algorithm}[h]
\caption{\texttt{OrthoProjDPP}: Sample from an orthogonal projection DPP}\label{alg:OrthoprojDPP}
\begin{algorithmic}
    \Function{\texttt{OrthoProjDPP}}{Y}\hspace{0.2cm}$\,\,\,\%\,\, Y\in\mathbb{R}^{n\times r}$ is orthogonal and $K = YY^T$
    \State sample $\gets$ empty vector
    \For {$i=1:r$}\\
    \hspace{1cm}Draw $j$ from $\{1,\dots,n\}$ with $\prob{$X=j$} = \verb|norm|(Y[j,:])^2$ \\
    \hspace{1cm}Add $j$ to sample \\
    \hspace{1cm}$Q\gets \verb|Householder|(Y[j,:])$\\
    \hspace{1cm}$Y \gets (YQ)[:,2:\text{end}]$
    \EndFor\\
    \hspace{0.5cm}\Return sample
    \EndFunction
\end{algorithmic}
\end{algorithm}

\par Let us for a moment restrict our projection matrix $K\in\mathbb{R}^{n\times n}$ to be an orthogonal projection matrix, so that we have $K = YY^T$ for some orthogonal (unitary, if complex) matrix $Y\in\mathbb{R}^{n\times r}$. The probability $P_i$ above is then equivalent to the squared norm of the $i^\text{th}$ row of $Y$, $\sum_{j=1}^r Y_{ij}^2$, divided by $r$. We multiply a Householder reflector $Q$ \cite{edelman18338} of the $i^\text{th}$ row of $Y$ on the right side of $Y$, so that $YQ$ has the $i^\text{th}$ row $(P_i, 0, \dots, 0)$. If we let $\tilde{Y} = YQ$ we have
\begin{equation*}
    \prob{$j\in\mathcal{J}|i\in\mathcal{J}$} = \sum_{k=2}^r \tilde{Y}_{jk}^2,
\end{equation*}
which means that the matrix $Z\in\mathbb{R}^{n\times (r-1)}$ obtained by deleting the first column of $\tilde{Y}$ (which is again orthogonal) plays the same role as $Y$ when drawing the first index point $i$. In other words, $ZZ^T$ is the rank $r-1$ marginal kernel of the DPP conditioned on the first sample index point $i$. Recursively applying this procedure $r$ times, we obtain Algorithm~\ref{alg:OrthoprojDPP} for orthogonal projection DPPs.

\par The algorithm introduced by Hough et al. in \cite{hough2006determinantal} samples from a Hermitian DPP using the fact that it is a mixture of projection DPPs via the eigendecomposition of $K$. Algorithm~\ref{alg:SymDPP} outlines this sampling algorithm. 

\begin{algorithm}[h]
\caption{\texttt{HermDPP}: Sample from a Hermitian DPP}\label{alg:SymDPP}
\begin{algorithmic}
    \Function{\texttt{HermDPP}}{$X, \Lambda$} \hspace{0.2cm}$\%\,\,$Eigendecomposition $K = X\Lambda X^T$, $K\in\mathbb{R}^{n\times n}$
    \State mask $\gets$ empty vector
    \For {$i=1:n$}
    \If{Bernoulli($\Lambda[i]$)==1}
    \State Add $i$ to mask 
    \EndIf 
    \EndFor
    \State $Y\gets X[:, \text{mask}]$ \\
    \hspace{0.5cm}\Return \verb|OrthoProjDPP(|$YY^T$\verb|)|
    \EndFunction
\end{algorithmic}
\end{algorithm}

\subsection{Sampling with conditional probabilites}\label{sec:condapproach}

\par Algorithm \ref{alg:SymDPP} and the subsequent select a single sample point of $\mathcal{J}$ at each time step. Thus, the set of points that are `not selected' is determined at the final step of the algorithm. On the other hand, some recent work \cite{launay2020exact,poulson2020high} uses a different approach based on conditional probabilities. At each step, rather than sampling an index point, we `observe' a single index point and determine whether it is going to be drawn or not. 

\par A central idea comes from the following block LU decomposition where the rows and columns are partitioned by $(m, n-m)$, 
\begin{equation*}
    K = \twotwo{I_m}{0}{K_{21}K_{11}^{-1}}{K_{22}-K_{21}K_{11}^{-1}K_{12}}\twotwo{K_{11}}{K_{12}}{0}{I_{n-m}}.
\end{equation*}
This is equivalent to the result of $m$ steps of the (unpivoted) LU decomposition. We have the following conditional probability for any subset $S$ of $\{m+1, \dots, n\}$,
\begin{equation*}
    \det\left(\left[(K_{22}-K_{21}K_{11}^{-1}K_{12})_{i,j}\right]_{i,j\in S}\right) = \prob{$S\subset\mathcal{J}$ $|$ $\{1, \dots, m\}\subset \mathcal{J}$}.
\end{equation*}
In other words, the matrix $K_{22}-K_{21}K_{11}^{-1}K_{12}$ (i.e., the Schur complement) serves as a new kernel for the DPP conditioned on all of $1,\dots,m$ already being drawn. Recall from the remark in the introduction that indices $1, \dots, m$ can be in fact any $m$ indices, since the order of observation can be arbitrarily chosen by some row/column permutation. Furthermore, a DPP with the condition that some index points are not selected can also be analyzed in a similar manner. The following Proposition summarizes this idea.
\begin{proposition}[\cite{poulson2020high}]\label{prop:poulsondiscrete}
Let $K$ be the kernel of the DPP $\mathcal{J}$. Given disjoint subsets $A, B$ of the ground set $\{1, \dots, n\}$, we have the following probabilities.
\begin{align*}
    \prob{$B\subset \mathcal{J}$ $|$ $A\subset\mathcal{J}$} &= \det(K_{B,B} - K_{B, A}K_{A,A}^{-1} K_{A, B}), \\
    \prob{$B\subset \mathcal{J}$ $|$ $A\subset \mathcal{J}^c$} &= \det(K_{B,B} - K_{B, A}(K_{A,A}-I)^{-1} K_{A, B}),
\end{align*}
where $K_{X,Y}$ is the submatrix of $K$ with row indices $X$ and column indices $Y$. 
\end{proposition}
Performing Bernoulli trials on pivots and using Proposition~\ref{prop:poulsondiscrete} we have the following Algorithm~\ref{alg:LUDPP} from \cite{launay2020exact,poulson2020high}. 

\begin{algorithm}[H]
\caption{\texttt{genDPP}: Sample from a general DPP}\label{alg:LUDPP}
\begin{algorithmic}
    \Function{\texttt{genDPP}}{$K$} \hspace{0.2cm}$\%\,\,K\in\mathbb{R}^{n\times n}$
    \State sample $\gets$ empty vector
    \For {$i=1:n$}
    \If{Bernoulli($K[i,i]$)==1}
    \State Add $i$ to sample
    \Else 
    \State $K[i,i]\gets K[i,i]-1$
    \EndIf
    \State $K[i+1:n,i+1:n]\,\,-=\,\,K[i+1:n,i]*K[i,i+1:n]/K[i,i]$
    \EndFor \\
    \hspace{0.5cm} \Return sample
    \EndFunction
\end{algorithmic}
\end{algorithm}

\par Although this ``greedy" type approach may be inefficient, there is one significant advantage: \textit{Algorithm~\ref{alg:LUDPP} enables sampling of a non-Hermitian DPP}. For example the discretized DPP of Dyson Brownian motion, which we will discuss in Section~\ref{sec:airyprocess}, is a non-Hermitian DPP. Another example is the Aztec diamond domino tiling \cite{jockusch1998random,johansson2002non,johansson2005arctic} which is used in \cite{poulson2020high} as an example of a non-Hermitian DPP.  

\begin{algorithm}[h]
\caption{\texttt{nonOrthoProjDPP}: Sample from a non-Hermitian projection DPP}\label{alg:HybridDPP}
\begin{algorithmic}
    \Function{\texttt{nonOrthoProjDPP}}{$K$} \hspace{0.2cm}$\%\,\,K\in\mathbb{R}^{n\times n}$ non-Hermitian rank $r$ projection
    \State sample $\gets$ empty vector
    \For {$i=1:r$}\\
    \hspace{1cm}Draw $j$ from $\{1,\dots,n\}$ with $\prob{$X=j$} = K[j,j]/(r-i+1)$ \\
    \hspace{1cm}Add $j$ to sample
    \State $K\,\,-=\,\,K[:,j]*K[j,:]/K[j,j]$
    \EndFor\\
    \hspace{0.5cm} \Return sample
    \EndFunction
\end{algorithmic}
\end{algorithm}

\par However one might notice that the DPP kernel of the Aztec diamond obtained from Kenyon's formula using the inverse Kasteleyn matrix \cite{chhita2015asymptotic,kenyon1997local} is non-Hermitian but is still a projection matrix. (A clue would be that any sample always has the fixed size $n(n+1)=\text{number of dominos}$.) When we have a non-Hermitian projection DPP with a small rank, Algorithm~\ref{alg:LUDPP} can be improved by combining with Algorithm~\ref{alg:OrthoprojDPP}. At each step we draw indices as we did in Algorithm~\ref{alg:OrthoprojDPP} from diagonal entries and then we modify the kernel as in Algorithm~\ref{alg:LUDPP}. This reduces the number of steps in Algorithm~\ref{alg:LUDPP} from the size of $K$ to its rank. We briefly sketch this hybrid approach in Algorithm~\ref{alg:HybridDPP}. 

\par We additionally note that unfortunately, the Aztec diamond DPP has its rank and size of the same order, which only obtains a small or even no improvement in practice. Nonetheless, if a kernel $K$ is non-Hermitian projection with $\text{rank}(K)\ll n$, theoretically Algorithm~\ref{alg:HybridDPP} should outperform Algorithm~\ref{alg:LUDPP}. Table \ref{tab:algchoice} is the summary of appropriate choices of exact DPP samplers, depending on the kernel $K$. 


\begin{table}[h]
    \centering
    \begin{tabular}{|c|c|c|c|}
    \hline
    Alg \# & Hermitian? & Projection? & Example \\ \hline
    \ref{alg:OrthoprojDPP} & $\color{green!45!black}\checkmark$ & $\color{green!45!black}\checkmark$ & Hermite kernel (finite GUE) \\
    \ref{alg:SymDPP} & $\color{green!45!black}\checkmark$ & {\color{red}x} & Airy kernel (soft-edge truncated) \\
    \ref{alg:LUDPP} & {\color{red}x} & {\color{red}x} & Airy process (Section~\ref{sec:airyprocess}) \\
    \ref{alg:HybridDPP} & {\color{red}x} & $\color{green!45!black}\checkmark$  & Aztec diamond \\ \hline
    \end{tabular}
    \vspace{-0.2cm}\caption{Choice of an exact sampling algorithm depending on properties of the kernel of a DPP.}\vspace{-0.3cm}
    \label{tab:algchoice}
\end{table}


\subsection{Application to Aztec diamonds: efficiently sampling a DR path}\label{sec:drpath}

\par Not only can Algorithm~\ref{alg:LUDPP} sample non-Hermitian DPPs but also it enables partial sampling. One application that benefits from the partial sampling of Algorithm~\ref{alg:LUDPP} is the sampling of \textit{DR paths} \cite{stanley1999enumerative} of the Aztec diamond. The \textit{north polar region} (NPR) boundary process is of a great interest due to its connection to corner growth, and eventually to the Airy process \cite{johansson2005arctic}. DR paths of the Aztec diamond are defined as follows: we label vertical dominos east (E) and west (W), according to their checkerboard patterns. A simple way to remember is that the vertical domino fits the westmost corner is a W-domino. Similarly we label horizontal dominos N-dominos and S-dominos. In Figure~\ref{fig:aztec_diamond}, red, blue, green, yellow dominos are W, E, N, S-dominos, respectively. We draw horizontal lines in the middle of S-dominos, $\pm\frac{\pi}{4}$ degree lines passing through the center of W, E-dominos, respectively. The right side of Figure~\ref{fig:aztec_diamond} has such lines drawn in red. It is known that these lines form $n$ continuous paths, which are called the DR paths.

\begin{figure}[h]
    \centering
    \frame{
    \includegraphics[width=0.45\textwidth]{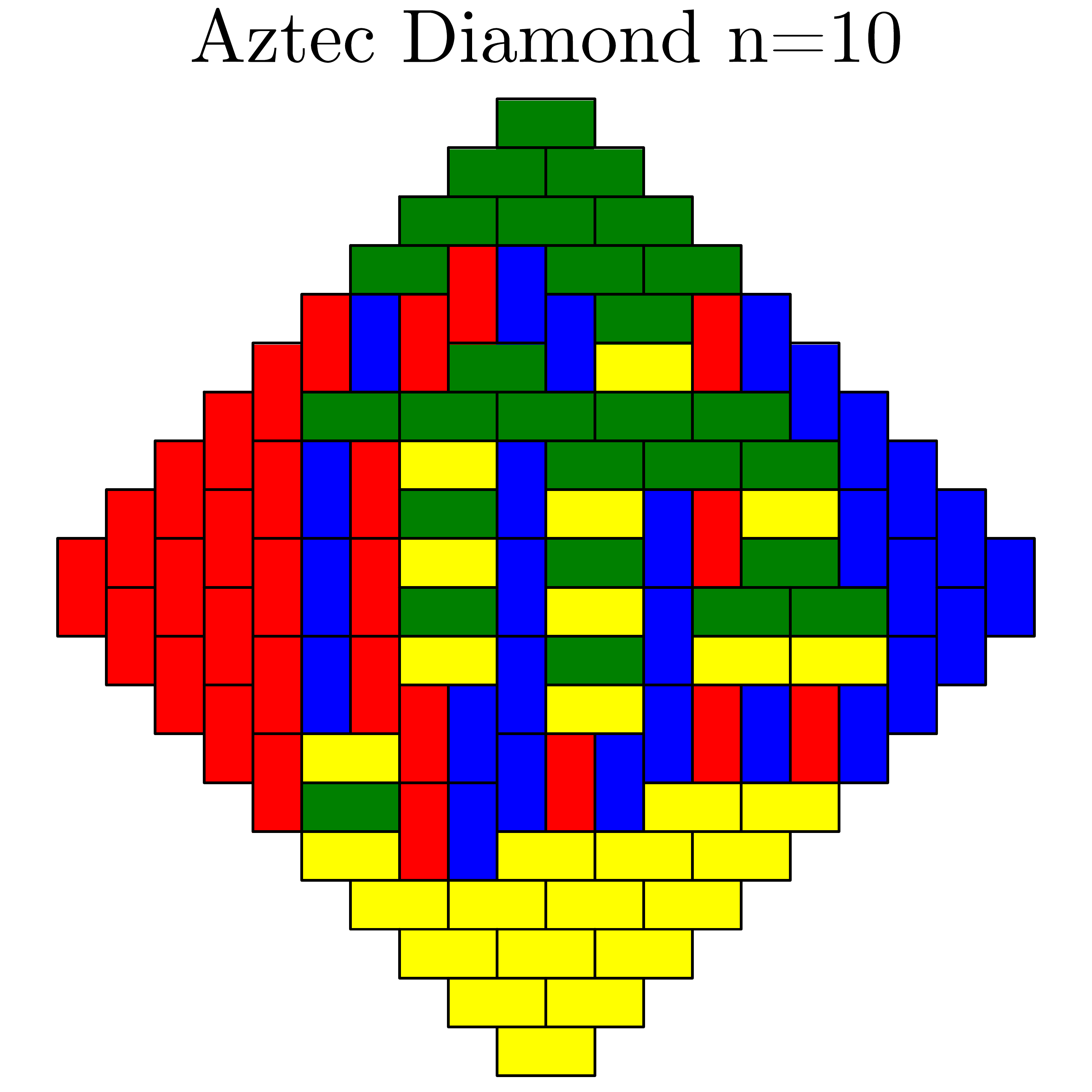}
    \includegraphics[width=0.45\textwidth]{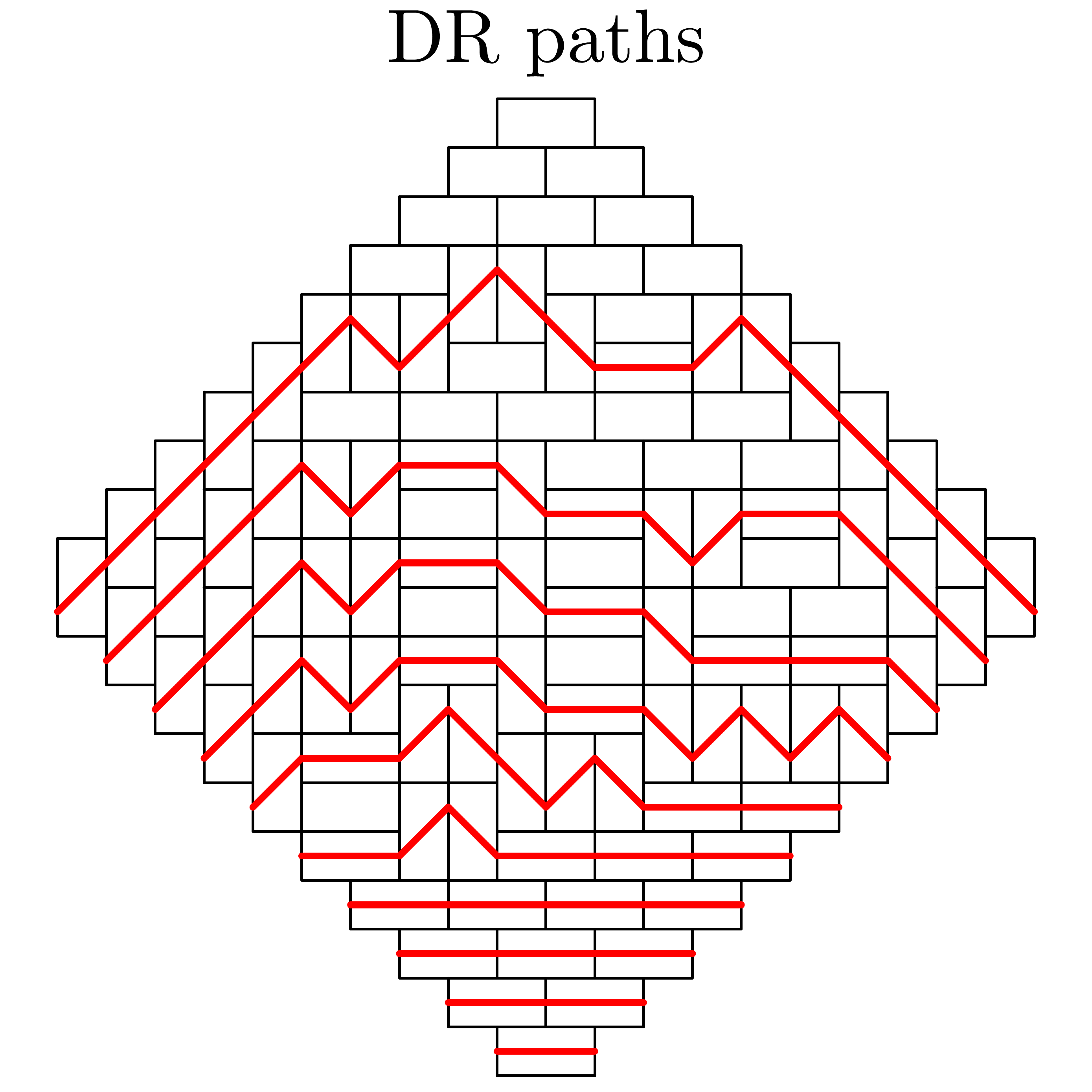}}
    \vspace{-0.2cm}\caption{An Aztec diamond domino tiling with $n=10$ sampled from a DPP (left) and corresponding DR paths in red (right)}
    \label{fig:aztec_diamond}
\end{figure}

\par An interesting observation is that we do not need the whole Aztec domino configuration to obtain the top DR path. Using Algorithm~\ref{alg:LUDPP} partially we can efficiently sample the top DR path (and similarly any $k^\text{th}$ DR path) by only \textit{observing} the possible dominos along the path. We first begin by observing the westmost corner. (Recall the last point of the remark in the introduction, that we can observe in any desired order in Algorithm~\ref{alg:LUDPP}.) For instance let us assume we sampled a W-domino there. Then we `observe' the three next possible (W, E, S) dominos that share the upper half of the right edge of the sampled W-domino, since they are the three possible extensions to the current path.  We iteratively do this until we reach the eastmost corner, which is the end of the top DR path. 

\par Complexity of sampling the whole Aztec diamond by Algorithm~\ref{alg:LUDPP} is $O(n^6)$ (LU decomposition), since the kernel size is $4n^2$. On the other hand, the partial sampling of the top DR path described above has $O(n^3)$ complexity, if we perform a dynamic memory allocation as we move along the path. (Notice in this case the `effective' kernel has size $O(n)$.) Figure~\ref{fig:aztec_diamond2} illustrates this point, where sampling only the top DR path is about 40 times faster. This partial sampling is comparable to the known algorithms, e.g., the shuffling algorithm \cite{elkies1992alternating}.

\begin{figure}[h]
    \centering
    \fbox{
    \begin{tabular}{p{0.9\textwidth}}
    \includegraphics[width=0.43\textwidth]{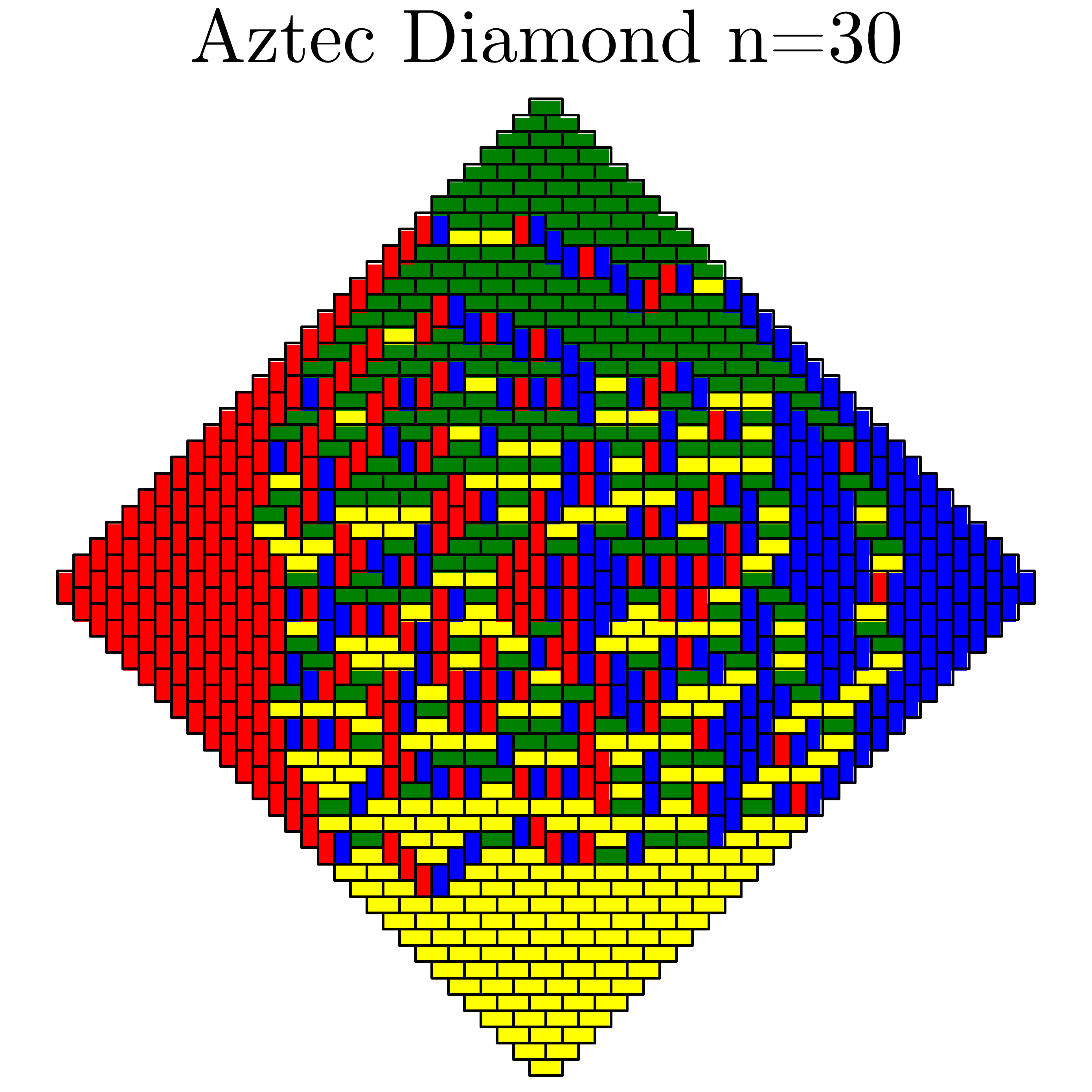}
    \includegraphics[width=0.43\textwidth]{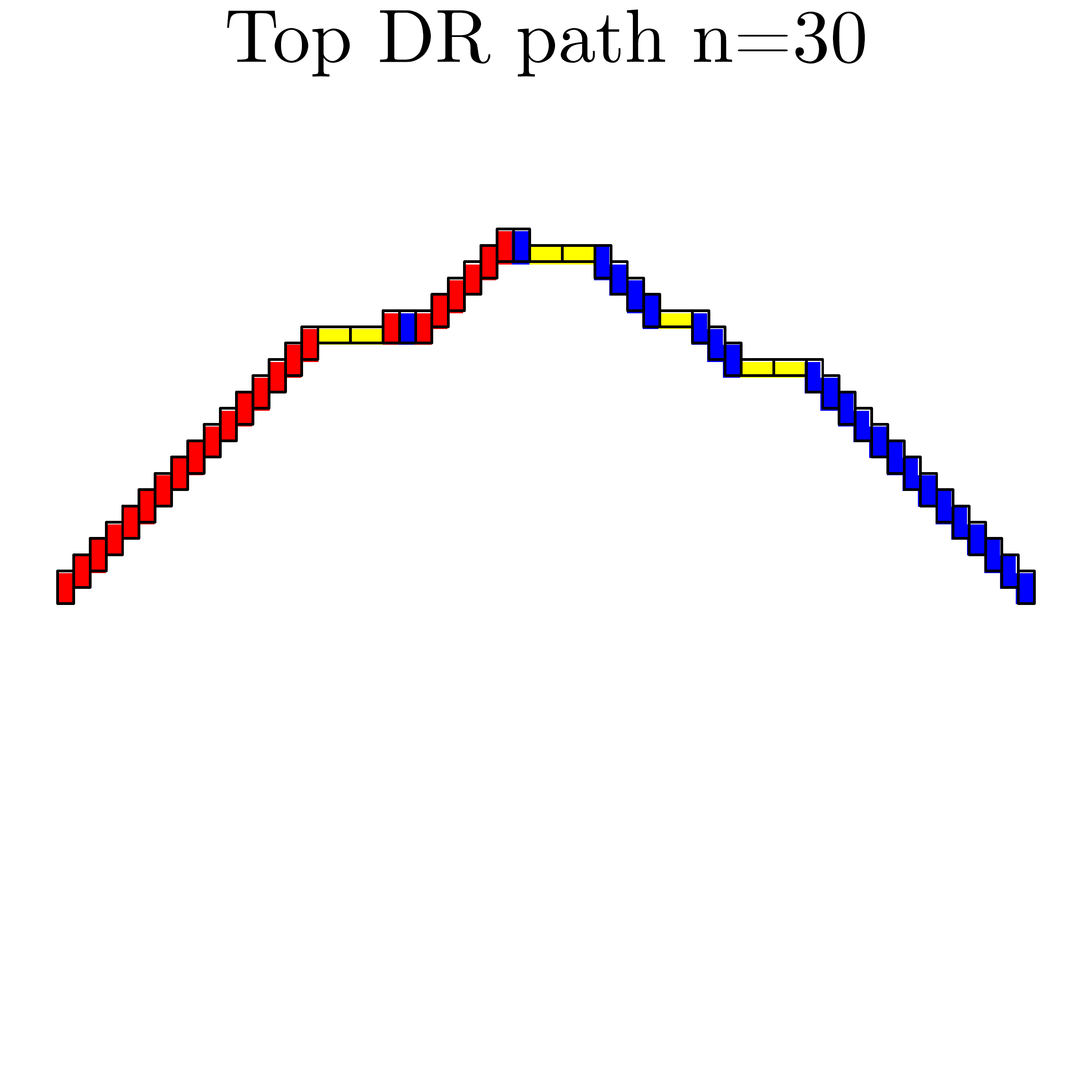} \\
    \centering DPP Sample time: 52.44 sec \hspace{1cm} DPP Sample time: 1.28 sec
    \end{tabular}}
    \caption{An $n=30$ Aztec diamond sampled from a DPP (left) and a top DR path sampled with partial sampling without the whole domino (right). The sampling times are 52.44 seconds and 1.28 seconds, respectively. With 50 seconds one can sample a top DR path with $n=50$.}
    \label{fig:aztec_diamond2}
\end{figure}

\subsection{Continuous DPPs}\label{sec:contDPPsampling}

\par Many concepts in discrete DPPs extend to continuous DPPs. The ground set $G$ now becomes a continuous interval (or any set) and the marginal kernel matrix becomes the kernel function $K:G\times G\to \mathbb{C}$. A continuous DPP defines the following $n$-point correlation function, which is the continuous analogue of $\prob{$S\subset \mathcal{J}$}$ in \eqref{eq:DPPdefinition},
\begin{equation*}
    p(x_1, \dots, x_n) = \det\left([K(x_i, x_j)]_{i, j=1, \dots, n}\right).
\end{equation*}
One way to describe the $n$-point correlation function $p(x_1, \dots, x_n)$ is the following:
\begin{equation}\label{eq:npointcordef}
    \lim_{\delta x\to 0}\frac{1}{(\delta x)^n}{\prob{$n$ points at length $\delta x$ intervals around $x_1, \dots, x_n$}}.
\end{equation}

\par Sampling algorithms for a continuous DPP could be generalized from the discrete case. For a continuous projection DPP with a bounded trace, i.e., 
\begin{equation*}
    \int K(x, y)K(y, z) dy = K(x, z)\hspace{0.5cm}\text{and}\hspace{0.5cm}\int K(x, x)dx < \infty,
\end{equation*}
Algorithm~\ref{alg:HybridDPP} works in a similar manner. The point selection at each timestep is now a univariate random variable with its PDF proportional to the diagonal $K(x, x)$. More details could be found in \cite{guillaume2020sampling}.

\par Another method for sampling from a continuous DPP is by discretizing a continuous DPP to a discrete DPP. Let us use the Hermite kernel $K_\text{Herm}^{(N)}$, \eqref{eq:HermiteKernel} as an example. To create a finite matrix we truncate the ground set. For the Hermite kernel with $N=5$, it is known that the largest eigenvalue lies around $\sqrt{2N}\approx 3$, and the probability that an eigenvalue lies outside a wide interval of length $L$, for example $(-10, 10)$, is already much lower than the double precision machine epsilon. Then with $M\approx L/\delta x$ length $\delta x$ intervals in the truncated region, create an $M\times M$ matrix $K$ such that $K_{ij} = K(x_i, x_j) \delta x$, where $x_j$ is the midpoint of the $j^\text{th}$ interval.\footnote{Of course, as $\delta x\to 0$ this does not have to be the midpoint.} From \eqref{eq:npointcordef} the determinants of principal submatrices approach probabilities of sample points lying around corresponding intervals, as $\delta x\to 0$. On the other hand if one tries to compute integrals such as $\tr K = \int K(x, x)dx$ or the Fredholm determinant $\det(I-K)$, one may use weights and points from quadrature rules, see \cite{bornemann2009numerical} for details. In Section~\ref{sec:airyprocess} we discretize the extended Airy kernel to sample Airy processes using Algorithm~\ref{alg:LUDPP}.

\section{The conditional DPP method applied to random matrix theory}\label{sec:mainresults}

Let $K$ be a kernel for an integral operator such that the $n$-point correlation function $p$ of some random matrix eigenvalues (or generally, any $\beta=2$ orthogonal polynomial ensemble) is given as an $n\times n$ determinant \eqref{eq:npointdeterminant}. Examples of such random matrices are the GUE, LUE, soft-edge, hard-edge scaling limit. As mentioned above, $p(x_1, \dots, x_n)$ is the continuous analogue of the set-level complementary cumulative distribution function (CCDF), $\prob{$S\subset \mathcal{J}$}$ in \eqref{eq:DPPdefinition}. 

\par The probability that no eigenvalue exists in a given interval $J$ plays an important role. In the continuous case, it is given as the following Fredholm determinant
\begin{equation}\label{eq:detImKcont}
    \prob{No eigenvalue in $J$} = \det(I - K\rt_J),
\end{equation}
which is a continuous generalization of \eqref{eq:detImK}. 

\subsection{Kernel of the DPP from conditional probabilities}\label{sec:detexpressions}
Let us first prove Proposition~\ref{prop:conditionalkernel}, which is the continuous analogue of Proposition~\ref{prop:poulsondiscrete}. 
\begin{proof}[Proof of Proposition~\ref{prop:conditionalkernel}]
One step of the LU decomposition (in reverse order) of the $(n+1)$-point correlation function is
\begin{equation*}
    [K(x_i, x_j)]_{i, j=1, \dots, n+1} = \twotwo{\left[K^{(x_{n+1})}(x_i, x_j)\right]_{i,j=1, \dots,n}}{v}{0}{1} \twotwo{I_n}{0}{w}{K(x_{n+1}, x_{n+1})},
\end{equation*}
where $v = [K(x_i, x_{n+1})]_{i=1, \dots, n}/K(x_{n+1}, x_{n+1})$ and $w^T = [K(x_{n+1}, x_i)]_{i=1, \dots, n}$. From the multiplicativity of determinants we get 
\begin{equation*}
    \det\left(\left[K(x_i, x_j)\right]_{i, j=1, \dots, n+1}\right) = \det\left(\left[K^{(x_{n+1})}(x_i, x_j)\right]_{i, j=1, \dots, n}\right)\cdot K(x_{n+1}, x_{n+1}).
\end{equation*}
Then, since $p^{(x_{n+1})}(x_1, \dots, x_n)=\det([K^{(x_{n+1})}(x_i, x_j)]_{i,j=1,\dots,n})$ and using \eqref{eq:npointcordef} for the left hand side and $K(x_{n+1},x_{n+1})$ we have
\begin{align*}
    p^{(x_{n+1})}&(x_1, \dots, x_n) \\
    &= \lim_{\delta x\to 0}\frac{1}{(\delta x)^n}\frac{\prob{Eigenvalues at $\delta x$ intervals around $x_1, \dots, x_{n+1}$}}{\prob{Eigenvalue at $\delta x$ interval around $x_{n+1}$}}, 
\end{align*}
which is the desired $n$-point correlation function conditioned on the event that an eigenvalue existing around an infinitesimal interval  containing  $x_{n+1}$. 
\end{proof}

\par Note that we can also generate a kernel for an $n$-point correlation function conditioned on the event that a fixed location \textit{does not} contain an eigenvalue, by replacing the denominator of the kernel \eqref{eq:newkernel} with  $K(s, s)-1$. However, we note that the condition that an eigenvalue does not exist at a specific location is less useful in our applications than the condition that an eigenvalue exists at a specific location, in the continuous setting. 

\par We generalize Proposition~\ref{prop:conditionalkernel} to a conditional DPP kernel with any number ($m$) of forced eigenvalues. 
\begin{proposition}\label{prop:conditionalkernelmpoints}
With the same assumptions as in Proposition~\ref{prop:conditionalkernel}, the $n$-point correlation function $p^{(s_1, \dots, s_m)}$ of the eigenvalues conditioned on the event that $m$ eigenvalues already exist around infinitesimal intervals around $s_1, \dots, s_m$ is
\begin{equation*}
    p^{(s_1,\dots,s_m)}(x_1, \dots, x_n) = \det\left(\left[K^{(s_1,\dots,s_m)}(x_i, x_j)\right]_{i,j=1,\dots,n}\right),
\end{equation*}
where the kernel $K^{(s_1,\dots,s_m)}$ is given as
\begin{align*}
    K&^{(s_1, \dots, s_m)}(x, y) \\&= K(x,y) - 
    \begin{bmatrix} K(x, s_1) \\ \vdots \\ K(x, s_m) \end{bmatrix}^T
    \!\!\begin{bmatrix} K(s_1, s_1) & \cdots\!\! & K(s_1, s_m) \\
     \vdots & \ddots\!\! & \vdots \\ 
     K(s_m, s_1) & \cdots\!\! & K(s_m, s_m) \end{bmatrix}^{-1}\!\!
    \begin{bmatrix} K(s_1, y) \\ \vdots \\ K(s_m, y) \end{bmatrix}.
    \end{align*}
\end{proposition}
\begin{proof} 
The proof is analogous to the proof of Proposition~\ref{prop:conditionalkernel} replacing the LU decomposition with the block LU decomposition of an $(n+m)\times (n+m)$ matrix, with row/column partitions $(n, m)$.  
\end{proof}

\subsection{The PDF of extreme (largest) eigenvalues}\label{sec:extremepdf}
One can derive a new expression for the PDF of the largest eigenvalue of a random matrix. Let us first review some basic eigenvalue statistics and conventions. For the $N\times N$ GUE we have 
\begin{equation*}
    E_2^{(N)}(0;J) := \prob{No eigenvalue of the $N\times N$ GUE is in $J$} = \det(I - K_\text{Herm}^{(N)}\rt_{J}),
\end{equation*}
where the Hermite kernel is defined in \eqref{eq:HermiteKernel}. In particular we denote $E_2^{(N)}(0;(0,s))$ simply by $E^{(N)}(0;s)$.

\par As $N\to \infty$ some scaling limits are defined. With the sine kernel $\Ks(x, y) = \frac{\sin\pi(x-y)}{\pi(x-y)}$ one obtains the bulk scaling limit\footnote{In fact, an appropriate scaling of the $\beta=2$ Laguerre and Jacobi ensembles in the bulk also yield the same bulk limit \cite{nagao1991correlation}.} defined with mean spacing 1,
\begin{equation*} 
    E(0;s) := \lim_{N\to\infty}E^{(N)}(0;\frac{\pi}{\sqrt{2n}}s) = \det\left(I - \Ks\rt_{(0, s)}\right).
\end{equation*}
Also with the Airy kernel
\begin{equation}\label{eq:airykernel}
    \Kai(x, y) = \frac{\Ai(x)\Ai'(y) - \Ai'(x)\Ai(y)}{x-y},
\end{equation}
we have the following Fredholm determinant representation of the soft-edge scaling limit of the GUE\footnote{As in the bulk, this can also be the largest eigenvalue of the LUE in the appropriate soft-edge scaling limit \cite{forrester1993spectrum}.} (the Tracy--Widom distribution $F_2$)
\begin{equation*}
    F_2(s) :=\!\lim_{N\to\infty}\!E_2^{(N)}\!\!\left(0; (\sqrt{2n}+\frac{s}{\sqrt{2n^{1/3}}}, \infty)\right)\!\! = \prob{$\lambda_\text{max}\leq s$} = \det\left(I - \Kai\rt_{(s, \infty)}\right).
\end{equation*}

\par Depending on the position and restrictions imposed on the eigenvalues, $\det(I-K)$ turns into several different distribution functions. At the edge (either hard or soft) $\det(I-K)$ becomes the CDF. In particular, at the $+\infty$ side of the soft-edge, $\det(I-\Kai\rt_{(s,\infty)})$ equals $\prob{$\lambda_{\text{max}} \leq s$}$, the CDF of the largest eigenvalue. Similarly, at the LUE hard-edge\footnote{Also the Jacobi ensemble with $\beta=2$ has the same hard-edge scaling limit \cite{borodin2003increasing}.} $\det(I-\Kbess{\alpha}\rt_{(0,s)})$ equals $\prob{$\lambda_{\text{min}}\geq s$}$, the CCDF.

\par On the other hand in the bulk, $\det(I-K)$ is not a CDF. Rather, its derivative is a CDF, i.e., $E(0;s)$ in the bulk has the following derivatives \cite[Chapter 6.1.2]{mehta2004random},
\begin{equation}\label{eq:F0p0def}
    \tilde{F}(0;s):= -\frac{d}{ds}E(0;s), \hspace{1cm} p(0;s) := -\frac{d}{ds}\tilde{F}(0;s).
\end{equation}
The first derivative $\tilde{F}(0;s)$ is the probability (in the bulk) that, for a randomly chosen level around zero, the distance $s$ to the right contains no eigenvalue. If we let a random variable $D$ be the distance (again, starting from any randomly chosen level) to the right until the next level, $\tilde{F}(0;s) = \mathbb{P}(D > s)$ is a CCDF. It follows that $p(0;s)$ is the PDF of $D$.

\par Some of these probabilities can alternatively be described by conditional probabilities. For example, the PDF of the smallest eigenvalue at the hard-edge is the product of (1) the $1$-point correlation function at $s$ ($=\Kbess{\alpha}(s,s)$) and (2) the probability that no eigenvalue lies in $(0, s)$, conditioned on an eigenvalue existing around $s$. The latter could be obtained from Proposition~\ref{prop:conditionalkernel} and becomes a new expression for the PDF of the smallest eigenvalue at the hard-edge,
\begin{equation}\label{eq:pdfhardedge}
    f_{\text{hard}, \alpha}(s) = \underbrace{\Kbess{\alpha}(s,s)}_{\text{(1) Level at $s$}}\hspace{0.2cm}\cdot\hspace{-0.4cm}\underbrace{\det\left(I - \Kbess{\alpha}^{(s)}\rt_{(0, s)}\right)}_{\text{(2) No levels on $(0,s)$, given level at $s$}}\hspace{-0.6cm}.
\end{equation}
In Proposition~\ref{prop:twpdf} we have already introduced the PDF of the Tracy--Widom distribution with the same idea. Generalizing these, we get the following Corollary:

\begin{corollary}\label{cor:derivativekernel}
Given a kernel $K$ as in Proposition~\ref{prop:conditionalkernel}, a kernel of a continuous DPP of some random matrix eigenvalues. Let 
\begin{equation*}
    f(a, b) = \det(I - K\rt_{(a, b)}),
\end{equation*}
the probability that no eigenvalue lies in $J=(a, b)$. Then, the following holds:
\begin{align}\label{eq:daderivative} 
    \frac{d}{da}f(a, b) &= K(a, a)\det(I - K^{(a)}\rt_{J}) = \prob{eigenvalue at $a$ and none in $J$}, \\ \label{eq:dbderivative}
    -\frac{d}{db}f(a, b) &= K(b, b) \det(I - K^{(b)}\rt_{J})= \prob{eigenvalue at $b$ and none in $J$},
\end{align}
with the kernels $K^{(a)}$ and $K^{(b)}$ defined as in \eqref{eq:newkernel}.
\end{corollary}
\begin{proof}
We prove \eqref{eq:daderivative} and the proof for \eqref{eq:dbderivative} is similar. Let $a' = a +\delta a$. 
\begin{align*}
    \frac{1}{\delta a}(&f(a', b)-f(a, b)) = \frac{1}{\delta a}\left(\prob{Nothing in $(a', b)$} - \prob{Nothing in $(a, b)$}\right)\\
    &= \delta a^{-1}\prob{Eigenvalue at $(a, a')$, none in $(a', b)$} \\
    &= \delta a^{-1}\prob{None in $(a', b) \,|$ Eigenvalue at $(a, a')$}\cdot \prob{Eigenvalue at $(a, a')$} \\
    &= \delta a^{-1}\det(I-K^{(a)}\rt_{(a', b)})\cdot \prob{Eigenvalue at $(a, a')$}
\end{align*}
As we let $\delta a\to 0$, $\delta a^{-1}\prob{Eigenvalue at $(a, a')$}$ goes to $K(a, a)$. 
\end{proof}
\par Let us also apply Corollary~\ref{cor:derivativekernel} to the distribution function \eqref{eq:F0p0def} in the bulk. For $\tilde{F}(0;s)$, since $\Ks(0,0) = 1$ we have
\begin{equation}\label{eq:F0fredholm}
    \tilde{F}(0;s) = -\det(I - \Ks^{(0)}\rt_{(0, s)}),
\end{equation}
where \eqref{eq:newkernel} defines $\Ks^{(0)} = \frac{\sin \pi(x-y)}{\pi(x-y)} - \frac{\sin\pi x \sin\pi y}{\pi^2xy}$. Applying the Corollary once more we get
\begin{equation}\label{eq:p0freholm}
    p(0;s) = \Ks^{(0)}(s,s)\cdot \det(I - \Ks^{(0)*}\rt_{(0,s)}),
\end{equation}
where the kernel $\Ks^{(0)*}(x, y) := \Ks^{(0)}(x,y)-\frac{\Ks^{(0)}(x,s)\Ks^{(0)}(s,y)}{\Ks^{(0)}(s,s)}$ is a twice derived sine kernel. Figure~\ref{fig:F0p0} is the plot of $\tilde{F}(0;s)$, $p(0;s)$ using \eqref{eq:F0fredholm}, \eqref{eq:p0freholm}, respectively. 

\begin{figure}[h]
    \begin{tabular}{cc}
    $\tilde{F}(0;s)$ & $p(0;s)$ \\
    \includegraphics[width=0.47\textwidth]{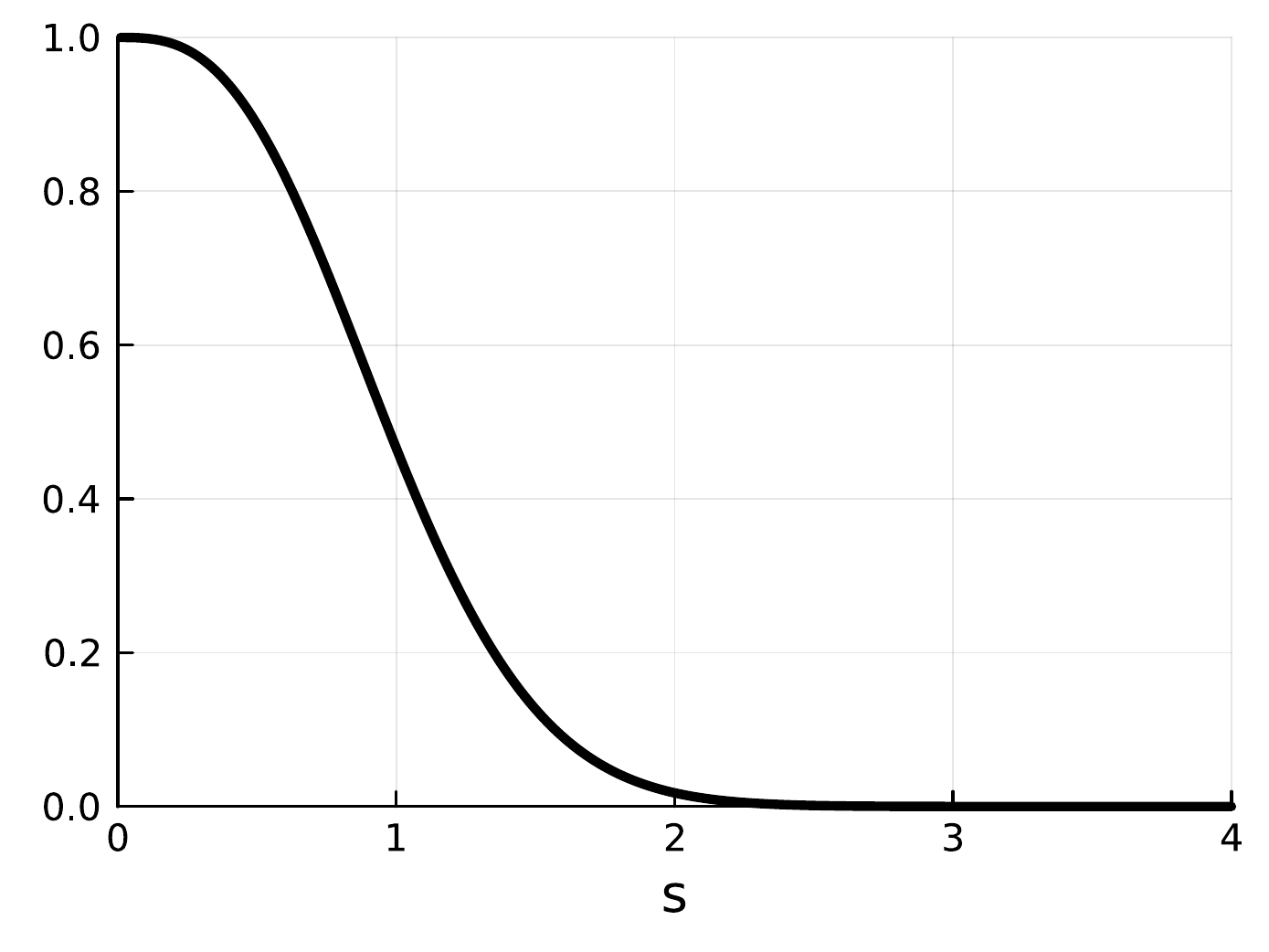} & \includegraphics[width=0.47\textwidth]{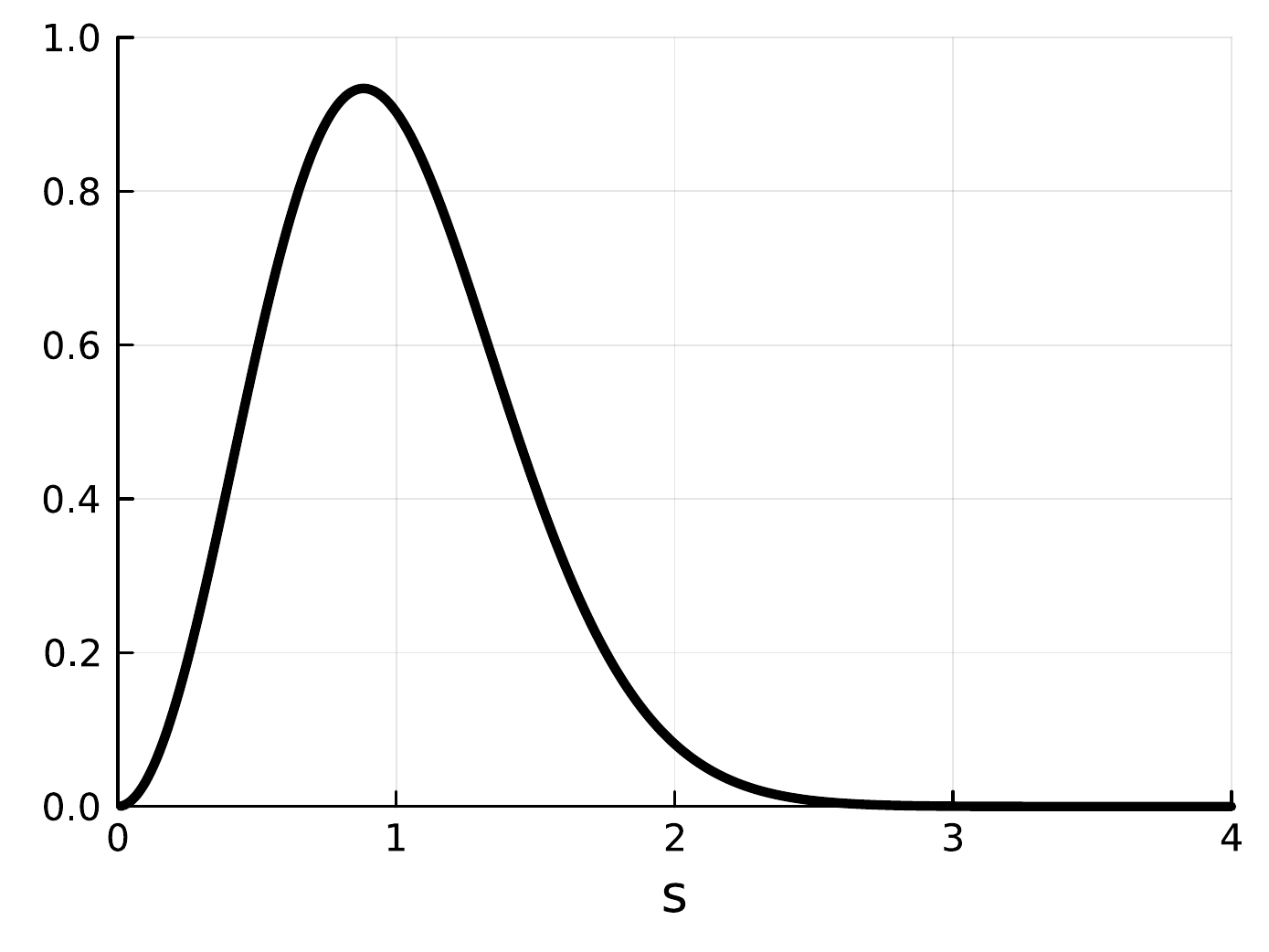}
    \end{tabular}
    \caption{Plots of $\tilde{F}(0;s)$ (left) and $p(0;s)$ (right) defined in \eqref{eq:F0p0def}. Numerical computation was done by evaluating Fredholm determinants \cite{bornemann2009numerical} in equations \eqref{eq:F0fredholm}, \eqref{eq:p0freholm}.\label{fig:F0p0}}
\end{figure}

\par Moreover, Table~\ref{tab:p0moments} is the first four moments of $D$ obtained from the computation of $p(0;D=s)$ values. See code \href{https://github.com/sw2030/RMTexperiments/blob/main/codes/F0p0.ipynb}{\texttt{F0p0}} for the implementation resulting in Figure~\ref{fig:F0p0} and Table~\ref{tab:p0moments}.
\begin{table}[h]
    \centering
    \begin{tabular}{|c|c|c|c|}
    \hline
    Mean & Variance & Skewness & Excess Kurtosis  \\ \hline
    1.0 & 0.179 993 877 691 8 & 0.497 063 620 491 8 & 0.126 699 848 039 9 \\ \hline
    \end{tabular}
    \caption{The first four moments of the spacing $D$ near zero in the bulk up to 13 digits. Total computation time is 0.052 second. One can compare this result to Table 8 of \cite{bornemann2010numerical}. \label{tab:p0moments}}
\end{table}

\par One benefit of our formulae, such as Proposition~\ref{prop:twpdf}, \eqref{eq:pdfhardedge}, \eqref{eq:F0fredholm}, and \eqref{eq:p0freholm}, is an efficient and accurate numerical computation. We provide numerical experiments that compare our computations of the Tracy--Widom PDF, \eqref{eq:tracywidompdf}, $f_2(s) = \Kai(s, s)\det(I - \Kai^{(s)}\rt_{(s, \infty)})$ with some other numerical approaches. Numerical results show that our expressions can be as efficient and potentially more accurate. 

\par One approach for computing the Tracy--Widom PDF is an automatic differentiation on $F_2(s)$ combined with the Fredholm determinant computation \cite{bornemann2009numerical} of $F_2(s) = \det(I-\Kai\rt_{(s,\infty)})$. Another approach has also been suggested recently in \cite[Eq (37b)]{bornemann2022stirling}, 
\begin{equation*}
    f_2(s) = -F_2(s)\cdot \tr\big((I-K)^{-1}K'\big)\rt_{L^2(s, \infty)} = F_2(s)\cdot \langle(I-K)^{-1} \Ai, \Ai\rangle_{L^2(s, \infty)},
\end{equation*}
where $K'$ is defined as $K'(x, y) = \left(\frac{\partial}{\partial x}+\frac{\partial}{\partial x}\right)K(x, y)$. See equation (34b) and nearby discussion in \cite{bornemann2022stirling}. In Table~\ref{tab:twpdfcompare}, we compare the accuracy of the computation of $f_2(s)$, for different numbers of quadrature points and $s$ values, using two different methods: (1) Equation (37b) of \cite{bornemann2022stirling} and (2) our approach, \eqref{eq:tracywidompdf}. Comparison for other random matrix statistics such as in the bulk (sine kernel) or hard-edge (Bessel kernel) for the values of the derivative of log determinant and comparison against the automatic differentiation could also be found in the code \href{https://github.com/sw2030/RMTexperiments/blob/main/codes/prime-computation.ipynb}{\texttt{prime-computation}}.

\begin{table}[h]
    \centering
    \begin{tabular}{r|cc|cc}
     & \multicolumn{2}{c|}{$m=10$} & \multicolumn{2}{c}{$m=20$}\\ 
    $s\,\,$ & Eq (37b) of \cite{bornemann2022stirling} & Eq \eqref{eq:tracywidompdf} & Eq (37b) of \cite{bornemann2022stirling} & Eq \eqref{eq:tracywidompdf}\\ \hline \hline
    -4.0 &  $3.79\times 10^{-1}$ & $2.76\times 10^{-2}$ & $3.72\times 10^{-7}$ &  $5.60\times 10^{-7}$  \\
    -3.5 & $2.94\times 10^{-2}$  & $6.13\times 10^{-2}$ & $5.14\times 10^{-8}$ &  $8.25\times 10^{-9}$ \\ 
    -3.0 & $2.48\times 10^{-2}$  & $1.68\times 10^{-2}$ & $1.07\times 10^{-8}$  & $4.18\times 10^{-9}$ \\
    -2.5 & $8.76\times 10^{-3}$ &  $1.89\times 10^{-3}$ & $9.41\times 10^{-11}$ & $1.39\times 10^{-10}$ \\
    -2.0 & $1.41\times 10^{-3}$ &  $2.55\times 10^{-3}$ & $1.28\times 10^{-10}$ & $2.11\times 10^{-13}$ \\
    -1.5 & $2.27\times 10^{-3}$ &  $1.70\times 10^{-3}$ & $4.70\times 10^{-11}$ & $3.33\times 10^{-12}$ \\
    -1.0 & $5.40\times 10^{-4}$ & $3.87\times 10^{-5}$  & $3.80\times 10^{-11}$ & $1.70\times 10^{-12}$ \\
    -0.5 & $5.61\times 10^{-4}$ & $3.06\times 10^{-6}$  & $7.12\times 10^{-12}$ & $1.87\times 10^{-14}$ \\
    0.0 & $5.92\times 10^{-4}$ & $1.17\times 10^{-6}$   & $1.15\times 10^{-11}$ & $9.35\times 10^{-14}$ \\
    \end{tabular}
    \caption{Relative errors of two approaches for computing the PDF of the Tracy--Widom distribution $f_2(s)$, with $m$ point Gauss-Legendre quadratures. The expression $f_2(s) = F_2(s)\cdot \langle(I-K)^{-1} \Ai, \Ai\rangle_{L^2(s, \infty)}$ is used from Equation (37b) of \cite{bornemann2022stirling} and our approach uses \eqref{eq:tracywidompdf}, $f_2(s) = \Kai(s, s)\cdot \det(I - \Kai^{(s)}\rt_{(s, \infty)})$. Our approach shows slightly better overall accuracy but the errors converge very quickly to the machine precision in both methods.}
    \label{tab:twpdfcompare}
\end{table}

\subsection{The PDF and CDF of the second largest eigenvalue}\label{sec:2ndeig}
Some new formulae for the second largest (second extreme) eigenvalue could be obtained from Proposition~\ref{prop:conditionalkernel}. Let us again take the soft-edge and the Airy kernel \eqref{eq:airykernel} as an example. A standard formula for the CDF of the $k^\text{th}$ largest eigenvalue is
\begin{equation}\label{eq:ktheigvalreccurence}
    F_2(k;s) = \sum_{m=0}^{k-1}\frac{(-1)^m}{m!}\frac{d^m}{dz^m}\det\left(I - z\Kai\rt_{(s, \infty)}\right)\big|_{z=1}.
\end{equation}
When $k=2$ we can derive a somewhat different formula for the CDF and PDF that does not involve differentiation or summation using the conditional DPP. For the CDF $\mathbb{P}(\lambda_2<s)$, we need to compute the probability that there is only one level lying in $(s, \infty)$ when $s$ is given. In the discrete DPP, the probability of having only one sample point is the trace of the $L$ kernel, where $L=(I-K)^{-1}K$, divided by $\det(I+L)$, which also holds similarly in the continuous case. Thus we obtain
\begin{equation*}
    F_2(2;s) = \underbrace{\tr\left((I-\Kai)^{-1}\Kai\right)\rt_{(s,\infty)}}_{\tr(L)}\cdot\underbrace{\det(I-\Kai\rt_{(s,\infty)})}_{\det(I+L)^{-1}}.
\end{equation*}

\par For the PDF of the second largest eigenvalue, we need to fix an eigenvalue at $s$ and proceed similarly as in the CDF. More precisely, we multiply (1) the $1$-point correlation function at $s$ and (2) the probability that there is only a single eigenvalue above $s$, conditioned on (1). That is,
\begin{align}\nonumber
    f_2(2;s) &= \frac{d}{ds}F_2(2;s) \\\label{eq:tw2ndpdf}
    &= \underbrace{\Kai(s, s)}_{\text{level at $s$}} \cdot 
    \underbrace{\tr\left((I-\Kai^{(s)})^{-1}\Kai^{(s)}\right)\rt_{(s,\infty)}\det(I-\Kai^{(s)}\rt_{(s,\infty)})}_{\text{Only one eigenvalue in $(s, \infty)$ given a level at $s$}}.
\end{align}
Generalizing this, we obtain the following proposition. 
\begin{proposition}\label{prop:2ndcdfpdf}
For a random matrix that has its eigenvalue $n$-point correlation function given as in Proposition~\ref{prop:conditionalkernel} with a kernel $K$, the CDF and PDF of the second largest eigenvalue $\lambda_2$ are
\begin{gather}\label{eq:2ndcdf}
    F^{\lambda_2}(s) = \tr\left((I-K)^{-1}K\right)\rt_{(s,\infty)}\cdot\det(I-K\rt_{(s,\infty)}), \\ \label{eq:2ndpdf}
    f^{\lambda_2}(s) = K(s, s) \cdot \tr\left((I-K^{(s)})^{-1}K^{(s)}\right)\rt_{(s,\infty)}\cdot\det(I-K^{(s)}\rt_{(s,\infty)}),
\end{gather}
where the kernel $K^{(s)}$ is given as \eqref{eq:newkernel}. 
\end{proposition}

\par Proposition~\ref{prop:2ndcdfpdf} could be used with the Airy kernel (soft-edge) as well as finite $N$ kernels like the Hermite kernel. In the hard-edge scaling the interval $(0, s)$ is used instead of $(s, \infty)$ with the Bessel kernel for the second smallest eigenvalue. Again, expressions \eqref{eq:2ndcdf} and \eqref{eq:2ndpdf} yield highly accurate numerical values using Bornemann's approach \cite{bornemann2009numerical}. Table \ref{tab:momentssecondeig} is the computed first four moments of the two smallest eigenvalues ($\lambda_1 < \lambda_2$) in the hard-edge and soft-edge using PDF formulas of Section~\ref{sec:extremepdf} and \eqref{eq:2ndpdf}.

\begin{table}[h]
    \centering
    \begin{tabular}{cc|r|r|r|r}
     & & Mean & Variance & Skewness & Kurtosis\\ \hline \hline
    \multirow{2}{*}{Soft} & $\lambda_1$ & -1.771 087 & 0.813 195 & 0.224 084 & 0.093 448 \\
    & $\lambda_2$ & -3.675 437 & 0.540 545 & 0.125 027 & 0.021 740 \\ \hline
    Hard & $\lambda_1$ & 4.000 000 & 16.000 000 & 2.000 000 & 6.000 000\\
    $\alpha=0$ & $\lambda_2$ & 24.362 715 & 140.367 319 & 0.924 147 & 1.225 112\\ \hline
    Hard & $\lambda_1$ & 10.873 127 & 55.745 139 & 1.320 312 & 2.541 266\\
    $\alpha=1$ & $\lambda_2$ & 40.812 203 & 259.898 510 & 0.737 801 & 0.764 990\\ \hline
    Hard & $\lambda_1$ & 20.362 715 &  124.367 319 & 1.015 815 & 1.461 306\\
    $\alpha=2$ & $\lambda_2$ & 60.112 814 & 416.851 440 & 0.622 605 & 0.532 483 
    \end{tabular}
    \caption{The first four moments (the last column is the excess kurtosis) of the two extreme eigenvalues, soft-edge and hard-edge. Computation time is less than a second for the whole table. See codes \href{https://github.com/sw2030/RMTexperiments/blob/main/codes/cor-coeff-softedge.ipynb}{\texttt{cor-coeff-softedge}} and \href{https://github.com/sw2030/RMTexperiments/blob/main/codes/cor-coeff-hardedge.ipynb}{\texttt{cor-coeff-hardedge}} for the implementation.}
    \label{tab:momentssecondeig}
\end{table}

\subsection{The joint distribution of the $k$ largest eigenvalues}\label{sec:klargest}

\par In this section we derive an expression for the joint distribution of the $k$ largest (or smallest) eigenvalues in terms of a Fredholm determinant.

\par In the case of $k=2$, the joint density of the two smallest eigenvalues of the Laguerre ensemble has been studied in \cite{forrester2007distribution}. In addition, the joint distribution of the two smallest eigenvalues at the hard-edge is obtained in terms of the solution of a differential equation that resembles Jimbo-Miwa-Okamoto $\sigma$-form of the Painlev{\'e} III. Analogously \cite{witte2013joint}, studies the joint density of the two largest eigenvalues at the soft-edge, obtaining an expression in terms of Painlev{\'e} II transcendents and isomonodromic components using hard-to-soft edge transition \cite{borodin2003increasing}. 

\par With Proposition~\ref{prop:conditionalkernelmpoints}, we derive an expression for the joint PDF of the $k$ largest eigenvalues in terms of the Fredholm determinant. Let us take for example $k=2$ and consider the $N\times N$ GUE. For simplicity let $K$ be the Hermite kernel \eqref{eq:HermiteKernel} and $p$ the $N\times N$ GUE eigenvalue $n$-point correlation function. The $n$-point correlation function after forcing (and conditioning on) two eigenvalues $x_1, x_2$ is 
\begin{equation*}
    p^{(x_1, x_2)}(y_1, \dots, y_n) = \det\left(\left[K^{(x_1, x_2)}(y_i, y_j)\right]_{i, j=1, \dots, n}\right),
\end{equation*}
where the kernel $K^{(x_1, x_2)}$ is given as
\begin{equation*}
    K^{(x_1, x_2)}(x, y) := K(x, y) - \twoone{K(x, x_1)}{K(x, x_2)}^T\twotwo{K(x_1, x_1)}{K(x_1, x_2)}{K(x_2, x_1)}{K(x_2, x_2)}^{-1}\twoone{K(x_1, y)}{K(x_2, y)}. 
\end{equation*}

Then, the joint PDF $f^{(\lambda_1, \lambda_2)}$ of the two largest eigenvalues $\lambda_1\geq \lambda_2$ is obtained as the following conditional probability argument assuming $x_1 > x_2$,
\begin{align*}
    f^{(\lambda_1, \lambda_2)}(x_1, x_2) &= p(x_1, x_2) \cdot \prob{No other eigenvalues in $(x_2, \infty)\,\,|\,\,\lambda_1=x_1, \lambda_2=x_2$}\\ 
    &= \det\left(\twotwo{K(x_1, x_1)}{K(x_1, x_2)}{K(x_2, x_1)}{K(x_2, x_2)}\right)\cdot \det(I - K^{(x_1,x_2)}\rt_{(x_2, \infty)}),
\end{align*}
and $f^{(\lambda_1, \lambda_2)}(x_1, x_2)$ vanishes when $x_1\leq x_2$ .

\par Similarly for the two smallest eigenvalues of the LUE, one can replace the interval $(x_2, \infty)$ of the above right hand side with $(0, x_2)$, where $0\leq x_1 \leq x_2$ being the two smallest eigenvalues. Again this approach works for any random matrix levels with determinantal $n$-point correlation function. 

\par Extending to general $k$'s we obtain an expression for the joint PDF of the $k$ largest eigenvalues, $f^{(\lambda_1, \dots, \lambda_k)}$, as the following.
\begin{proposition}\label{prop:klargesteig}
For a random matrix whose eigenvalue $n$-point correlation function is given in determinantal representation \eqref{eq:npointcordef} with a kernel $K$, we have the following joint PDF of the $k$ largest eigenvalues $\lambda_1 \geq \dots \geq \lambda_k$
\begin{equation*}
    f^{(\lambda_1, \dots, \lambda_k)}(x_1, \dots, x_k) = \det\left(\left[K(x_i,x_j)\right]_{i, j=1, \dots, k}\right)\cdot\det\left(I-K^{(x_1, \dots, x_k)}\rt_{(x_k, \infty)}\right),
\end{equation*}
for $x_1> \dots > x_k$ and vanishes otherwise, where the kernel $K^{(x_1, \dots, x_k)}(x, y)$ is defined as,
\begin{equation*}
    K^{(x_1, \dots, x_k)}(x, y) \!=\! K(x,y) - 
    \begin{bmatrix} K(x, x_1) \\ \vdots \\ K(x, x_n) \end{bmatrix}^T
    \!\!\begin{bmatrix} K(x_1, x_1) & \!\!\!\!\cdots & \!\!\!\!K(x_1, x_n) \\
     \vdots & \!\!\!\!\ddots & \!\!\!\!\vdots \\ 
     K(x_n, x_1) & \!\!\!\!\cdots & \!\!\!\!K(x_n, x_n) \end{bmatrix}^{-1}\!\!
    \begin{bmatrix} K(x_1, y) \\ \vdots \\ K(x_n, y) \end{bmatrix}.
\end{equation*}
\end{proposition}
In the following sections we give some examples of applications of Proposition~\ref{prop:klargesteig} with numerical experiments. Furthermore the first row of Figure~\ref{fig:gallery} contains some visualizations of these eigenvalue statistics that could be obtained from the joint PDF formula of $k=2$ and $k=3$ extreme eigenvalues.

\subsubsection{Correlation of the two extreme eigenvalues}\label{sec:correlation}
\par In \cite{bornemann2010numerical} the correlation coefficient $\rho(\lambda_1, \lambda_2)$ of the two largest eigenvalues at the soft-edge is computed from matrix valued kernels and a generating function approach similar to \eqref{eq:ktheigvalreccurence}. On the other hand, we can compute $\rho(\lambda_1, \lambda_2)$ with the following steps:
\begin{enumerate}\setlength\itemsep{0.4em}
    \item Compute $\mathbb{E}\lambda_1\lambda_2$ using the joint PDF $f^{(\lambda_1, \lambda_2)}(x_1, x_2)$ from Proposition~\ref{prop:klargesteig} and 2-dimensional Gauss quadrature on a (truncated) triangular region. 
    \item Compute $\mathbb{E}\lambda_1$, $\mathbb{E}\lambda_2$, $\sigma\lambda_1$, $\sigma\lambda_2$ using the PDF expressions \eqref{eq:tracywidompdf}, \eqref{eq:tw2ndpdf} and the Gauss-Legendre quadrature. Infinite intervals such as $(s, \infty)$ are handled by the strategy described in \cite[Section 7]{bornemann2009numerical}. 
    \item Compute $\rho(\lambda_1, \lambda_2) = (\mathbb{E}\lambda_1\lambda_2-\mathbb{E}\lambda_1\mathbb{E}\lambda_2)/(\sigma\lambda_1\sigma\lambda_2)$.
\end{enumerate}
The total computing time for obtaining 11 accurate digits, $\rho(\lambda_1, \lambda_2) = 0.50564723159$, is 118 seconds. Previously reported computing time is 16 hours \cite{bornemann2010numerical}. In addition to the soft-edge, we compute correlation coefficients of the two smallest eigenvalues at the hard-edge for $\alpha=0,1,2$. See Table \ref{tab:hardedgestats} for the result. 

\begin{table}[h]
    \centering
    \begin{tabular}{c|c}
    & $\rho(\lambda_1, \lambda_2)$ \\ \hline \hline
    Soft-edge & 0.505 647 231 59 \\ \hline
    Hard-edge $\alpha=0$ & 0.337 619 085 22 \\ \hline   
    Hard-edge $\alpha=1$ & 0.391 735 693 02 \\ \hline
    Hard-edge $\alpha=2$ & 0.417 187 915 41 \\ \hline
    \end{tabular}
    \caption{Computed values of correlation coefficients of the two largest eigenvalues at the soft-edge and the two smallest eigenvalues ($\lambda_1 < \lambda_2$) at the hard-edge scaling limit. Computation time is 117 seconds for the soft-edge and 139 seconds for the whole hard-edge correlation coefficients. See codes \href{https://github.com/sw2030/RMTexperiments/blob/main/codes/cor-coeff-softedge.ipynb}{\texttt{cor-coeff-softedge}} and \href{https://github.com/sw2030/RMTexperiments/blob/main/codes/cor-coeff-hardedge.ipynb}{\texttt{cor-coeff-hardedge}} for the detailed implementation.}
    \label{tab:hardedgestats}
\end{table}

\subsubsection{First eigenvalue spacing}\label{sec:spacing}

We compute moments of the distance between the first two eigenvalues (the first eigenvalue spacing) by computing the PDF and CDF using Proposition~\ref{prop:conditionalkernel}. In the second row of Figure~\ref{fig:gallery} we also plot some distributions of the first eigenvalue spacing, using the expressions we derive in this section.

\par The probability that the first spacing is at least $d$ can be obtained by integrating over $x$ the probability density of no further eigenvalues existing in $(x-d, \infty)$ given a level at $x$. From Proposition~\ref{prop:conditionalkernel} such a probability density is given as (for example, at the soft-edge)
\begin{equation*}
    \Kai(x, x)\cdot \det\left(I - \Kai^{(d)}\rt_{(x-d,\infty)}\right),
\end{equation*}
thus we obtain the CDF $G(s)$ of the first spacing,
\begin{equation*}
    G(d) = 1 - \int_\mathbb{R} \Kai(x, x)\det(I - \Kai^{(d)}\rt_{(x-d, \infty)})dx.
\end{equation*}
Moreover, simply using the joint PDF $f^{(\lambda_1, \lambda_2)}$ of the two eigenvalues computed above we obtain
\begin{equation}\label{eq:spacingpdf}
    A(d) = \int_\mathbb{R} d\cdot f^{(\lambda_1, \lambda_2)}(x, x-d)dx,
\end{equation}
which is the PDF of the first spacing. For the implementation we use Gauss-Legendre quadrature for the integration. 

\begin{table}[h]
    \centering
    \begin{tabular}{|c|c|c|c|}
    \hline
    Mean & Variance & Skewness & Excess Kurtosis  \\ \hline
    1.904 350 489 721 & 0.683 252 055 105 & 0.562 291 976 040 & 0.270 091 960 715 \\ \hline
    \end{tabular}
    \caption{The first four moments of the distance between the first two eigenvalues of the soft-edge scaling, up to 12 digits. Equation \eqref{eq:spacingpdf} is used to compute these values. \label{tab:spacingmoments}}
\end{table}

\par Table~\ref{tab:spacingmoments} is the computed first four moments of the first eigenvalue spacing, up to 12 digits, with a total runtime of 236 seconds. These moments were already computed in \cite{witte2013joint} up to $5\sim 9$ digits, with a reported computing time of 5 hours. With \eqref{eq:spacingpdf} computation of the moments up to 9 digits needs 29 seconds of a runtime. Values of $A, G$ of the soft-edge scaling are verified up to 8 digits against Table 2 of \cite{witte2013joint}, with a total computing time of 354 seconds. Computation of $A, G$ values and comparison with previously known values could be found in \href{https://github.com/sw2030/RMTexperiments/blob/main/codes/first-spacing.ipynb}{\texttt{first-spacing}}.

\subsection{Sampling the Airy process and Dyson Brownian motion using DPP}\label{sec:airyprocess}

\par In this section we add an additional (time) parameter $t$ as a random matrix changes through time according to some stochastic processes. A random matrix diffusion or Dyson process, for example Dyson Browian motion (GUE diffusion), is another example of a DPP in random matrix theory. 
A multitime correlation function for the $N\times N$ GUE diffusion is given in terms of a block matrix determinant with the \textit{extended Hermite kernel} $K$ \cite{tracy2004differential},
\begin{equation}\label{eq:multitimecor}
    p(x_{t_s,i_s}\,\,; \,\,s=1, \dots, n \text{ and } i_s = 1, \dots, m_s)
    = \det\left(\left[K_{j,k}\right]_{j,k = 1, \dots, n}\right),
\end{equation}
where $K_{j,k}$ is an $m_j \times m_k$ matrix,
\begin{equation*}
    K_{j,k} = \left[K(x_{t_j, i_j}, x_{t_k, i_k})\right]_{\substack{i_j = 1, \dots, m_j\\ i_k = 1, \dots, m_k}}.
\end{equation*}

\par This is essentially the density of eigenvalues of the random matrix stochastic process at times $\{t_s\}_{s = 1, \dots, n}$ each existing on positions $x_{t_s, 1},\dots, x_{t_s, m_s}$. The determinantal multitime correlation function \eqref{eq:multitimecor} also holds for the soft-edge scaling, LUE, and other orthogonal polynomial ensembles with appropriately computed kernels as proved in \cite{eynard1998matrices}. Indeed the block matrix determinant can be discretized to a block matrix $K$ kernel for a DPP as prescribed in Section~\ref{sec:contDPPsampling}. 

\par In particular, with the \textit{extended Airy kernel}
\begin{equation*}
    K_{s, t}^\text{ext}(x, y) = \left\{\begin{array}{cc}
        \int_0^\infty e^{-\lambda(s-t)}\Ai(x+\lambda)\Ai(y+\lambda)d\lambda & \text{if }s\geq t \\
        -\int_{-\infty}^0 e^{-\lambda(s-t)}\Ai(x+\lambda)\Ai(y+\lambda)d\lambda & \text{if } s< t
        \end{array}\right.,
\end{equation*}
one has the multitime correlation function of the soft-edge scaling limit by \eqref{eq:multitimecor}. A special interest lies in the largest eigenvalue of this process and is called the \textit{Airy$_2$ process}, or just simply, the \textit{Airy process}. Obviously extending \eqref{eq:detImK}, the largest eigenvalue process is described by the following Fredholm determinant,
\begin{equation}\label{eq:FdetAiryprocess}
    \prob{$\cA(t_1)\leq s_1, \dots, \cA(t_n)\leq s_n$} = \det\left(I - K\rt_{L^2(s_1,\infty)\oplus\cdots\oplus L^2(s_n, \infty)}\right),
\end{equation}
where $K$ is the block kernel given as
\begin{equation*}
    K = \begin{bmatrix}
    K_{t_1, t_1}^\text{ext} & \cdots & K_{t_1, t_n}^\text{ext} \\ 
    \vdots & \ddots & \vdots \\ 
    K_{t_n, t_1}^\text{ext} & \cdots & K_{t_n, t_n}^\text{ext}
    \end{bmatrix}.
\end{equation*}

\par The Airy process is related to a number of applications, including the polynuclear growth process \cite{johansson2003discrete,prahofer2002scale}, the NPR boundary process of the Aztec diamond domino tiling \cite{johansson2005arctic} (i.e., the top DR path as $n\to \infty$ in Section~\ref{sec:drpath}), totally asymmetric simple exclusion process (TASEP), corner growth process \cite{johansson2003discrete}, and eventually related to the KPZ universality class with the narrow wedge initial condition. 

\par An example of numerical experiments on the Airy process is \cite{bornemann2009numerical}, where Bornemann uses the $2\times 2$ block Fredholm determinant \eqref{eq:FdetAiryprocess} to compute the two-point covariance $\text{cov}(\mathcal{A}(0), \mathcal{A}(t))$ and compare those values against its large $t$ and small $t$ expansions. 

\begin{figure}[h]
    \centering
    \boxed{
    \includegraphics[width=0.95\textwidth]{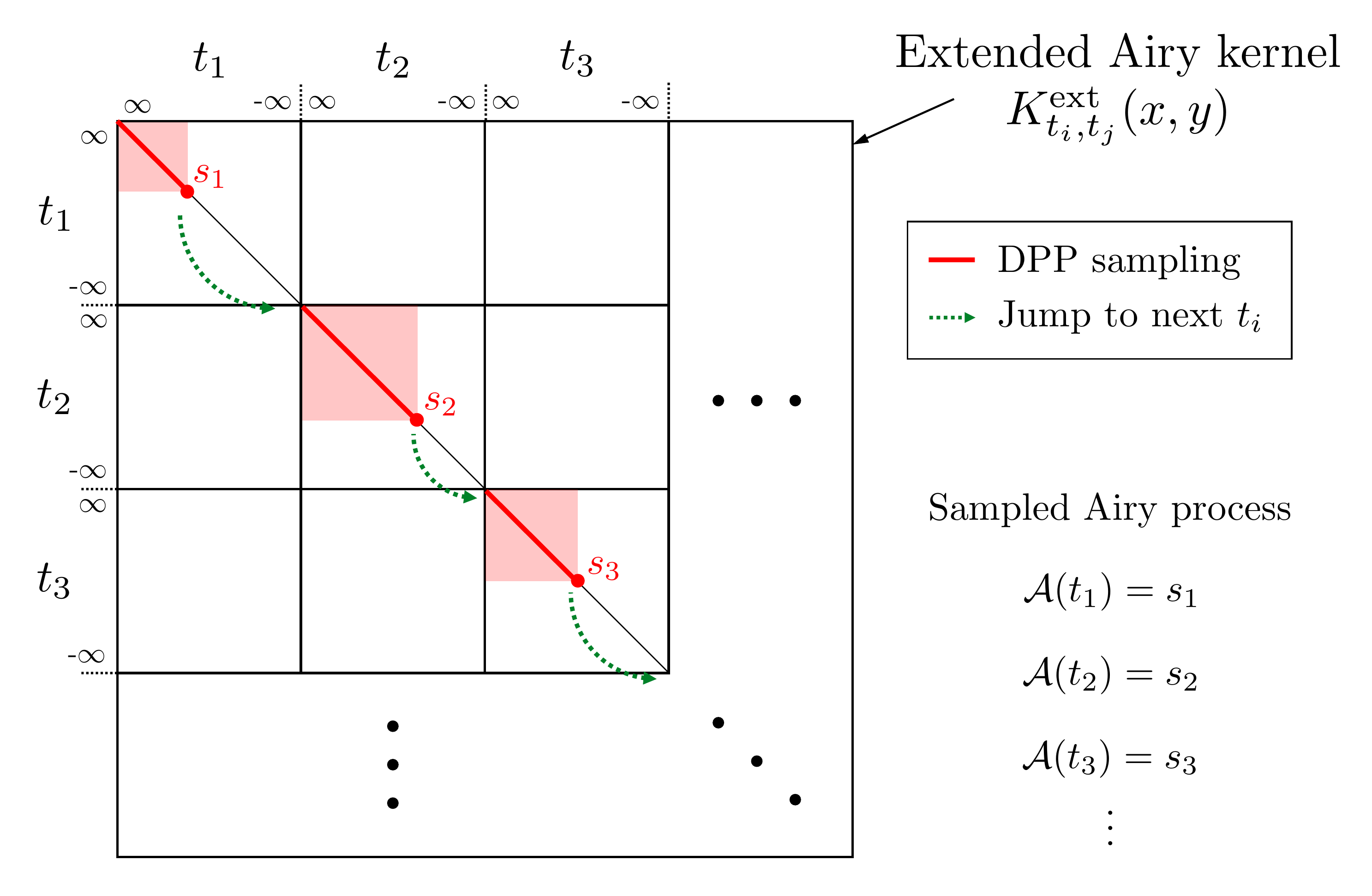}}
    \caption{Sketch of the sampling technique used to sample the Airy process. Each block represents discretized $K_{t_i, t_j}^\text{ext}$. At each timestep $t_i$ we perform Algorithm~\ref{alg:LUDPP}, starting from $+\infty$, until we hit the first (largest) eigenvalue. Then we jump to the next timestep, disregarding (not observing) the rest of eigenvalues at the current timestep. Note that each diagonal block of the kernel is the usual Airy kernel \eqref{eq:airykernel}.}
    \label{fig:airyprocesstechnique}
\end{figure}

\par Here we demonstrate another numerical experiment on the Airy process: Sampling Airy processes through multitime correlation function \eqref{eq:multitimecor} and the DPP sampler. We used a $201\times 201$ block matrix for sampling. The extended Airy kernel is non-Hermitian and non-projection, which can only be sampled by Algorithm~\ref{alg:LUDPP}. A simple multistep (multitime) modification of Algorithm~\ref{alg:LUDPP} is handy; In each timestep $t_i$, we proceed from $+\infty$ to $-\infty$, and jump to next block (next timestep $t_{i+1}$) when we find a largest eigenvalue at the current timestep. See Figure~\ref{fig:airyprocesstechnique} for the illustration of this algorithm. More details such as discretization and truncating the interval to get finite kernel could be found in Section~\ref{sec:contDPPsampling}.

\par In the left of Figure~\ref{fig:airyprocess} is the plot of five samples of the Airy process, sampled from the DPP defined by multitime correlation function of the (soft-edge scaled) GUE diffusion, i.e., the extended Airy kernel. We truncated the eigenvalue space to $[-5.0, 2.5]$, as probabilities that a sample from the Tracy--Widom distribution (which is the stationary distribution of the Airy process) larger than 2.5 and smaller than -5 are both around $2\times 10^{-5}$. The truncated eigenvalue domain is then discretized into 150 intervals, finally yielding $30150\times 30150$ kernel for the DPP.

\begin{figure}[h]
    \centering
    \includegraphics[width=0.49\textwidth]{airyprocess.pdf}
    \includegraphics[width=0.49\textwidth]{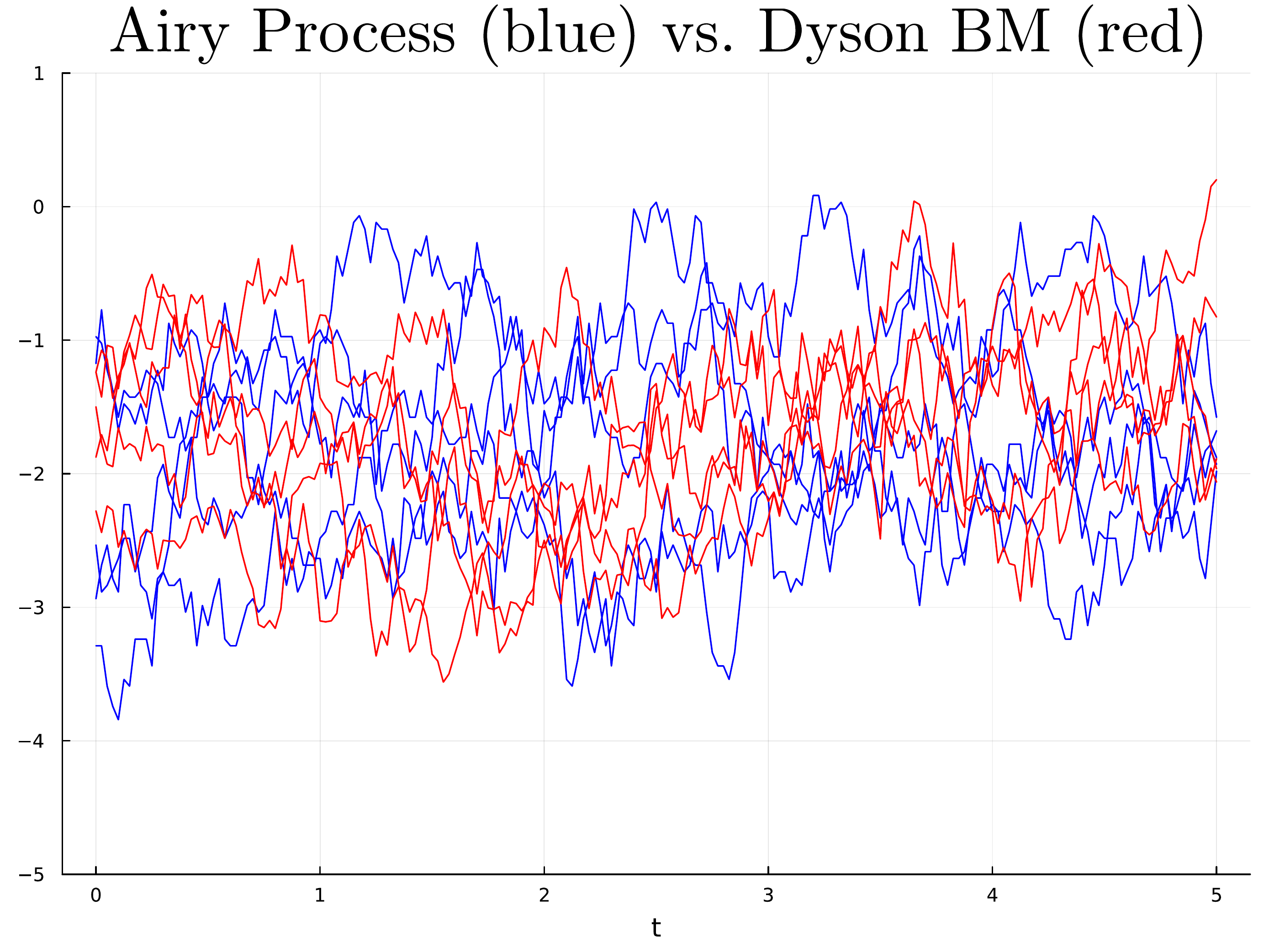}
    \captionsetup{singlelinecheck=off}
    \caption[.]{Left: Five samples of the Airy process from $t=0$ to $t=5$, sampled with $dt = 0.025$. Sampling time for a single sample is around 4 hours. Right: We additionally draw five samples of the largest eigenvalue process of $N=200$ Dyson Brownian motion, allowing the eye to observe the asymptotic behavior \eqref{eq:dbmtoairy}.}
    \label{fig:airyprocess}
\end{figure}

\par The Airy processes in Figure~\ref{fig:airyprocess} are non-asymptotic in a sense that they are not large $N$ asymptotics where $N\to \infty$ is the Airy process. It is known that $\lambda_{\text{max}}(t)$ of Dyson Brownian motion (GUE diffusion) recentered and rescaled according to
\begin{equation}\label{eq:dbmtoairy}
    \sqrt2 N^{\frac{1}{6}}\left(\lambda_\text{max}(N^{-\frac{1}{3}}t) - \sqrt{2N}\right),
\end{equation}
converges to the Airy process as $N\to\infty$. Samples of such large $N$ approximation are drawn in red in the right side of Figure~\ref{fig:airyprocess} by numerically simulating GUE diffusions (not using multitime extended Hermite kernel DPP).

\subsection*{Numerical experiment details}
All codes mentioned in the paper can be found online: \href{https://github.com/sw2030/RMTexperiments}{\texttt{https://github.com/sw2030/RMTexperiments}}. For numerical experiments (except Section~\ref{sec:airyprocess}) discussed in this work, we used a single core of an Apple M1 Pro CPU. For the Airy process sampling discussed in Section~\ref{sec:airyprocess}, we used 64 cores from four Xeon P8 CPUs for computing the DPP kernel and a single core of Xeon P8 CPU for sampling, from the MIT Supercloud server \cite{reuther2018interactive}. 

\subsection*{Acknowledgements}
We thank Jack Poulson for helpful comments, and especially pointing out that the dynamic sampling of the top DR path is in fact $O(n^3)$. We thank Dimitris Konomis and Aviva Englander for early versions of Dyson Brownian motion simulation that were begun as class projects. 
\vspace{0.5cm}

\setstretch{0.4}
\noindent
{\footnotesize This material is based upon work supported by the National Science Foundation under Grant No. DMS-1926686. The authors acknowledge the MIT SuperCloud and Lincoln Laboratory Supercomputing Center for providing HPC resources that have contributed to the research results reported within this paper. This material is based upon work supported by the National Science Foundation under grant no. OAC-1835443, grant no. SII-2029670, grant no. ECCS-2029670, grant no. OAC-2103804, and grant no. PHY-2028125. We also gratefully acknowledge the U.S. Agency for International Development through Penn State for grant no. S002283-USAID. The information, data, or work presented herein was funded in part by the Advanced Research Projects Agency-Energy (ARPA-E), U.S. Department of Energy, under Award Number DE-AR0001211 and DE-AR0001222. We also gratefully acknowledge the U.S. Agency for International Development through Penn State for grant no. S002283-USAID. The views and opinions of authors expressed herein do not necessarily state or reflect those of the United States Government or any agency thereof. This material was supported by The Research Council of Norway and Equinor ASA through Research Council project ``308817 - Digital wells for optimal production and drainage''. Research was sponsored by the United States Air Force Research Laboratory and the United States Air Force Artificial Intelligence Accelerator and was accomplished under Cooperative Agreement Number FA8750-19-2-1000. The views and conclusions contained in this document are those of the authors and should not be interpreted as representing the official policies, either expressed or implied, of the United States Air Force or the U.S. Government. The U.S. Government is authorized to reproduce and distribute reprints for Government purposes notwithstanding any copyright notation herein.}

\setstretch{1.0}
\bibliographystyle{plain}
\bibliography{bibliography}

\end{document}